%% file: APPT_ArXiv_final.tex
\documentclass[11pt]{article}

\pdfoutput=1

\usepackage{geometry}
\usepackage{booktabs} 
\usepackage[ruled]{algorithm2e} 

\SetAlFnt{\small}
\SetAlCapFnt{\small}
	\SetAlCapNameFnt{\small}
\SetAlCapHSkip{0pt}
\IncMargin{-\parindent}

\usepackage{amsthm}
\usepackage{amsmath,amssymb}

\usepackage{xcolor}
\usepackage[ocgcolorlinks]{hyperref} 
\usepackage{cleveref}
\usepackage{hyperref}
\usepackage{paralist}
\usepackage{wasysym}

\usepackage[]{graphicx}
\usepackage{grffile}
\usepackage{caption}
\usepackage{subcaption}

\newcommand{\R}{\mathbb{R}}

\usepackage[ruled]{algorithm2e} 

\SetAlFnt{\small}
\SetAlCapFnt{\small}
\SetAlCapNameFnt{\small}
\SetAlCapHSkip{0pt}
\IncMargin{-\parindent}

\colorlet{DarkRed}{red!50!black}
\colorlet{DarkGreen}{green!50!black}
\colorlet{DarkBlue}{blue!70!black}

\hypersetup{
	linkcolor = DarkRed,
	citecolor = DarkGreen,
	urlcolor = DarkBlue,
	bookmarksnumbered = true,
	linktocpage = true
}

\theoremstyle{plain}
\newtheorem{theorem}{Theorem}[section]
\newtheorem{fact}[theorem]{Fact}
\newtheorem{remark}[theorem]{Remark}
\newtheorem{lemma}[theorem]{Lemma}
\newtheorem{definition}[theorem]{Definition}

\newtheorem{claim}[theorem]{Claim}

\begin{document}
\title{Secretary and Online Matching Problems \\ with Machine Learned Advice}

\author{\begin{tabular}{cc}
	\begin{tabular}{c}
 Antonios Antoniadis\footnote{Work done in part while the
 author was at Saarland University and Max-Planck-Institute for
 Informatics and supported by DFG grant AN 1262/1-1.}  \\ 
 Universit\"at zu K\"oln\  \\
 \texttt{antoniadis@cs.uni-koeln.de} \\
 \quad \\
 \quad
	\end{tabular} & 
	\begin{tabular}{c}
	Themis Gouleakis\\ 
	Max Planck Institute for Informatics\\
	Saarland Informatics Campus\\
 \texttt{tgouleak@mpi-inf.mpg.de}\\
 \quad
	\end{tabular}\\
	\begin{tabular}{c}
	Pieter Kleer\\ 
	Max Planck Institute for Informatics\\  
	Saarland Informatics Campus\\
	\texttt{pkleer@mpi-inf.mpg.de}
	\end{tabular} & 
	\begin{tabular}{c}
		Pavel Kolev\\ 
	Max Planck Institute for Informatics\\  
	Saarland Informatics Campus\\
	\texttt{pkolev@mpi-inf.mpg.de}
	\end{tabular}
\end{tabular}}

\maketitle           
\begin{abstract}
The classical analysis of online algorithms, due to its  worst-case nature,
  can be quite pessimistic when the input instance at hand is
  far from worst-case. Often this is not an issue with machine
  learning approaches, which shine in exploiting patterns in past inputs
  in order to predict the future.  However, such predictions, although
  usually accurate, can be arbitrarily poor.
   Inspired by a recent line of work, we augment three well-known online settings
  with machine learned predictions 
  about the future, and develop algorithms that take them
  into account. In particular, we study the following online selection
  problems: (i) the classical secretary problem, (ii) online bipartite
  matching and (iii) the graphic matroid secretary problem. 
  Our algorithms still come with a worst-case performance
  guarantee in the case that predictions are subpar while obtaining an improved competitive ratio
  (over the best-known classical online algorithm for each problem)
  when the predictions are sufficiently accurate. For each algorithm, we
  establish a trade-off between the competitive ratios obtained in the two respective cases.    
\end{abstract}

\input{introduction.tex}

\input{preliminaries.tex}

\input{secretary.tex}

\input{matching.tex}
\input{graphic_matroid.tex}
\input{conclusion.tex}

\bibliographystyle{plain}
\bibliography{ref}

\appendix
\input{appendix_secretary.tex}

\newpage
\input{appendix_kesselheim.tex}

\newpage
\input{appendix_deterministic_graphic_matroid.tex}

\newpage
\input{appendix_graphic_matroid_secretary_predictions.tex}

\newpage
\input{truthful_transversal.tex}

\end{document}

%% file: introduction.tex
\section{Introduction} 
There has been enormous progress in the field of machine learning in
the last decade, which has affected a variety of other areas as well. One of these areas is the design of online algorithms. Traditionally, the analysis of such algorithms involves worst-case guarantees, which can often be quite pessimistic. It is conceivable though, that having prior knowledge regarding the online input (obtained using machine learning)  could potentially improve those guarantees significantly.  

In this work,
we consider various online selection algorithms augmented with so-called \emph{machine learned advice}. In particular, we consider secretary and online bipartite  matching problems. The high-level idea is to incorporate
some form of \emph{predictions} in an existing online algorithm in
order to get the best of two worlds: 
(i) provably improve the algorithm's performance guarantee in the case that predictions are sufficiently good,
while (ii) losing only a constant factor of the algorithm's existing 
worst-case performance guarantee, when the predictions are subpar. 
Improving the performance of classical online algorithms with the help of machine
learned predictions is a relatively new area that has gained a lot of attention 
in the last couple of years \cite{KraskaBCDP18, KhalilDNAS17,Rohatgi20,
  LattanziLMV20,LykourisV18,PurohitSK18,GollapudiP19,Mitzenmacher20,
  MedinaV17}.

We motivate the idea of incorporating such machine-learned advice,
in the class of problems studied in this work,
by illustrating a simple real-world problem.
Consider the following setting for selling a laptop on an online
platform.\footnote{This example is similar to the classical secretary
 problem \cite{G1960}.} Potential buyers arrive one by one, say, in a
uniformly random order, and report a price that they are willing to pay
for the laptop. Whenever a buyer arrives, we have to irrevocably
decide if we want to sell at the given price, or wait for a better
offer. Based on historical data, e.g., regarding previous online sales
of laptops with similar specs, the online platform might suggest a (machine learned)
prediction for the maximum price that some buyer is likely to offer
for the laptop.

How can we exploit this information in our decision process? 
One problem that arises here is that we do not have any formal guarantees for how accurate
the machine-learned advice is for any particular instance.
For example, suppose we get a prediction of $900$ dollars 
as the maximum price that some buyer will
likely offer. One extreme policy is to blindly trust this prediction 
and wait for the first buyer to come along that offers a price sufficiently 
close to $900$ dollars. 
If this prediction is indeed accurate, this policy has
an almost perfect performance guarantee, in the sense that we will
sell to the (almost) highest bidder. However, if the best offer is
only, say, $500$ dollars, we will never sell to this buyer 
(unless this offer arrives last), since the advice is to wait for a better offer to come along.
In particular, the performance guarantee of this selling policy 
	depends on the \emph{prediction error} ($400$ dollars in this case) 
	which can become arbitrarily large. 
The other extreme policy is to completely ignore the prediction of
$900$ dollars and just run the classical secretary algorithm: Observe
a $1/e$-fraction of the buyers, and then sell to the first buyer that
arrives afterwards, who offers a price higher than the best offer seen
in the initially observed fraction. This yields, in expectation, a
selling price of at least $1/e$ times the highest offer
\cite{L1961,D1962}.

Can we somehow combine the preceding two extreme selling-policies, 
so that we get a performance guarantee strictly better than that of $1/e$ 
in the case where the prediction for the highest offer is not too far off, 
while not loosing too much over the guarantee of $1/e$ otherwise? 
Note that (even partially) trusting poor predictions often comes at a price, 
and thus obtaining a competitive ratio worse than $1/e$ seems inevitable 
in this case. We show that there is in fact a trade-off between the competitive ratio that we can achieve when the prediction is accurate and the one we obtain when the prediction error turns out to be large.

\subsection{Our models and contributions}
\label{sec:contributions}
We show how one can incorporate predictions in various online
selection algorithms for problems that generalize the classical
secretary problem. The overall goal is to include \emph{as little
  predictive information as possible} into the algorithm, while still
  obtaining improvements in the case that the information is
  accurate. Our results are parameterized by (among other
parameters) the so-called prediction error $\eta$ that measures the
quality of the given predictions. 

We briefly sketch each of the problems
studied in this work, and then conclude with the description of the
meta-result in Theorem~\ref{thm:meta}, that applies to all of them.

\paragraph{Secretary problem.}
In order to illustrate our ideas and techniques, we start by augmenting the classical secretary problem\footnote{
	To be precise, we consider the so-called \emph{value maximization} 
	version of the problem, see Section \ref{sec:secretary} for details.}
with predictions. For details, see Section~\ref{sec:secretary}.
Here, we are given a prediction $p^*$ for 
the maximum value among all arriving secretaries.\footnote{
	This corresponds to a prediction for the maximum price somebody is 
	willing to offer in the laptop example.}  
The prediction error is then defined as $\eta = |p^* - v^*|$,
where $v^*$ is the true maximum value among all secretaries. We emphasize that the algorithm is not aware of the prediction error $\eta$,
	and this parameter is only used to analyze the algorithm's performance guarantee.

\paragraph{Online bipartite matching with vertex arrivals.}
In Section \ref{sec:bipartite}, we study the online bipartite matching problem 
in which the set of nodes $L$ of a bipartite graph $G = (L \cup R, E)$, 
with $|L| = n$ and $|R| = m$, 
arrives \emph{online} in a uniformly random order \cite{KP2009,KRTV2013}. 
Upon arrival, a node reveals the edge weights to its neighbors in $R$. 
We have to irrevocably decide if we want to match up the arrived online node 
with one of its (currently unmatched) neighbors in $R$. 
Kesselheim et al. \cite{KRTV2013} gave a tight $1/e$-competitive deterministic algorithm for this setting that significantly generalizes the same guarantee for the classical secretary algorithm \cite{L1961,D1962}. 

The prediction that we consider in this setting is 
a vector of values $p^* = (p_1^*,\dots,p_{m}^*)$
that predicts the edge weights adjacent to the nodes $r \in R$ 
in some fixed optimal (offline) bipartite matching. 
That is, the prediction $p^*$ indicates the existence of a 
fixed optimal bipartite matching in which 
each node $r \in R$ is adjacent to an edge with weight $p_r^*$. 
The prediction error is then the maximum prediction error taken over all nodes in $r \in R$ and minimized over all optimal matchings. This generalizes the prediction used for the classical secretary problem. This type of predictions  closely corresponds to the \emph{vertex-weighted online bipartite matching problem} \cite{AGKA2011}, which will be discussed in Section \ref{sec:bipartite}.

\paragraph{Graphic matroid secretary problem.}
In Section \ref{sec:graphic}, 
we augment the graphic matroid secretary problem with predictions. In this problem, the edges of a given undirected graph $G = (V,E)$, 
with $|V| = n$ and $|E| = m$, arrive in a uniformly random order. 
The goal is to select a subset of edges of maximum weight under the constraint that this subset is a forest. That is, it is not allowed to select a subset of edges that form a cycle in $G$. The best known algorithm for this online problem is a (randomized) $1/4$-competitive algorithm by Soto, Turkieltaub and Verdugo \cite{STV2018}. Their algorithm proceeds by first selecting no elements from a prefix of the sequence of elements with randomly chosen size, followed by selecting an element if and only if it belongs to a (canonically computed) offline optimal solution, and can be added to the set of elements currently selected online. This is inspired by the algorithm of Kesselheim et al. \cite{KRTV2013} for online bipartite matching.

As a result of possible independent interest, we show that there exists a \emph{deterministic} $(1/4 - o(1))$-competitive algorithm for the graphic matroid secretary problem, which can roughly be seen as a deterministic version of the algorithm of Soto et al. \cite{STV2018}. Alternatively, our algorithm can be seen  as a variation on the (deterministic) algorithm of Kesselheim et al. \cite{KRTV2013} for the case of online bipartite matching (in combination with an idea introduced in \cite{BIKK2018}).

The prediction that we consider here is a vector of values $p =
(p_1^*,\dots,p_n^*)$ where $p_i^*$ predicts the maximum
edge weight that node $i \in V$ is adjacent to, in the graph $G$. This is
equivalent to saying that $p_i^*$ is the maximum edge weight adjacent
to node $i \in V$ in a given optimal spanning tree (we assume that $G$ is
connected for sake of simplicity), which is, in a sense, in line with the predictions used in Section \ref{sec:bipartite}.  (We note that the
	predictions model the optimal spanning tree in the case when all
	edge-weights are pairwise distinct. Otherwise, there can be many
	(offline) optimal spanning trees, and thus the predictions do not
	encode a unique optimal spanning tree. We intentionally chose not to
	use predictions regarding which edges are part of an optimal
	solution, as in our opinion, such an assumption would be too strong.)
The prediction error is also defined similarly.

\subsubsection*{Meta Result}
We note that for all the problems above, one cannot hope for 
an algorithm with a performance guarantee better than $1/e$ 
in the corresponding settings without predictions, 
as this bound is known to be optimal for 
the classical secretary problem (and also applies to the other problems) \cite{L1961,D1962}. Hence, our goal is to design algorithms that improve upon
	the $1/e$ worst-case competitive ratio in the case where the prediction error 
	is sufficiently small, and otherwise (when the prediction error is large)
	never lose more than a constant (multiplicative) 
	factor over the worst-case competitive ratio.

For each of the preceding three problems, we augment existing algorithms 
  with predictions. All of our resulting algorithms are \emph{deterministic}.

In particular, we show that the canonical approaches 
for the secretary problem \cite{L1961,D1962} and 
the online bipartite matching problem \cite{KRTV2013} 
can be naturally augmented with predictions.
We also demonstrate how to adapt our novel deterministic algorithm 
for the graphic matroid secretary problem. 
Further, we comment on randomized approaches in each respective section.

\begin{theorem}[Meta Result]\label{thm:meta}
There is a polynomial time deterministic algorithm 
that incorporates the predictions $p^*$
such that for some constants $0 < \alpha, \beta < 1$ it is
\begin{enumerate}[i)]
\item $\alpha$-competitive with $\alpha > \frac{1}{e}$,
when the prediction error is sufficiently small; and
\item $\beta$-competitive with $\beta < \frac{1}{e}$, independently of the prediction error.
\end{enumerate}
\end{theorem}

\noindent We note that there is a correlation between the constants $\alpha$ and $\beta$,
which can be intuitively described as follows:
The more one is willing to give up in the worst-case guarantee,
i.e. the more confidence we have in the predictions,
the better the competitive ratio becomes in the case where the predictions are sufficiently accurate.

We next give a high-level overview of our approach. 
We split the sequence of uniformly random arrivals in three phases. 
In the first phase, we observe a fraction of the input without selecting anything. In the remaining two phases, we run two extreme policies,
which either exploit the predictions or ignore them
completely. 
Although, each of the aforementioned extreme policies can be analyzed individually, 
using existing techniques, it is a non-trivial task to show that
when combined they do not obstruct each other too much. In particular,
the execution order of these policies is crucial for the analysis. 

We give detailed formulations of the meta-result in Theorem \ref{thm:classical} for the secretary problem; in Theorem \ref{thm:bipartite} for the online bipartite matching problem; and in Theorem \ref{thm:graphic} for the graphic matroid secretary problem. 

In addition, we show that the online bipartite matching algorithm, given in Section \ref{sec:bipartite}, 
can be turned into a truthful mechanism 
in the case of the so-called \emph{single-value unit-demand domains}.
Details are given in Theorem \ref{thm:truthful} in Appendix \ref{app:truthful}. 
We show that the algorithm of Kesselheim et al. \cite{KRTV2013} can be turned into a truthful mechanism in the special case of \emph{uniform edge weights}, where for every fixed online node in $L$, there is a common weight on all edges adjacent to it. In addition, we show that truthfulness can be preserved when predictions are included in the algorithm.  
We note that Reiffenh\"auser \cite{R2019} recently gave a truthful mechanism for the general online bipartite matching setting (without uniform edge weights). It would be interesting to see if her algorithm can be augmented with predictions as well.

\begin{remark}
In the statements of Theorem \ref{thm:classical}, \ref{thm:bipartite} and \ref{thm:graphic} 
it is assumed that the set of objects $O$ arriving online (either vertices or edges) 
is asymptotically large. We hide $o(1)$-terms at certain places 
for the sake of readability.
\end{remark}

Although the predictions provide very little information about the
optimal solution,
we are still able to obtain improved theoretical guarantees 
in the case where the predictions are sufficiently accurate. 
In the online bipartite matching setting with predictions for the
nodes in $R$, we can essentially get close to a $1/2$-approximation
-- which is best possible -- assuming the predictions are close to
perfect. This follows from the fact that in this case we obtain the
so-called \emph{vertex-weighted online bipartite matching problem} for
which there is a deterministic $1/2$-competitive algorithm, and no
algorithm can do better \cite{AGKA2011}. Roughly speaking, our
algorithm converges to that in \cite{AGKA2011} when the predictions
get close to perfect. This will be discussed further in Section
\ref{sec:bipartite}.  For the graphic matroid secretary problem, we
are also able to get close to a $1/2$-approximation in the case where the
predictions (the maximum edge weights adjacent to the nodes in the
graph) get close to perfect. 
We note that this is probably not tight. We suspect that, when given perfect predictions, it is possible to obtain an algorithm with a better approximation guarantee. This is an interesting open problem.

\subsection{Related work}
This subsection consists of three parts. 
	First we discuss relevant approximation algorithms for the matroid secretary problem 
	without any form of prior information, 
	then we consider models that incorporate additional information,
	such as the area of prophet inequalities. 
	Finally, we give a short overview of related problems 
	that have been analyzed with the inclusion of machine learned advice 
	following the frameworks in  \cite{LykourisV18,PurohitSK18}, 
	which we study here as well.

\paragraph{Approximation algorithms for the matroid secretary problem.}
The classical secretary problem was originally introduced by Gardner \cite{G1960}, and solved by Lindley \cite{L1961} and Dynkin \cite{D1962}, who gave $1/e$-competitive algorithms. Babaioff et al. \cite{BIKK2018} introduced the matroid secretary problem, a considerable generalization of the classical secretary problem, where the goal is to select a set of secretaries with maximum total value under a matroid constraint for the set of feasible secretaries. They provided an $O(1/\log(r))$-competitive algorithm for this problem, where $r$ is the rank of the underlying matroid. Lachish \cite{L2014} later gave an  $O(1/\log\log(r))$-competitive algorithm,
and a simplified algorithm with the same guarantee was given by Feldman, Svensson and Zenklusen \cite{FSZ2014}. It is still a major open problem if there exists a constant-competitive algorithm for the matroid secretary problem. 
Nevertheless, many constant-competitive algorithms 
are known for special classes of matroids, and we mention those relevant 
to the results in this work (see, e.g., \cite{BIKK2018,STV2018} for further related work).

Babaioff et al. \cite{BIKK2018} provided a $1/16$-competitive
 algorithm for the case of transversal matroids, which was later
 improved to a $1/8$-competitive algorithm by Dimitrov and Plaxton
 \cite{DP2012}. Korula and P{\'a}l \cite{KP2009} provided the first
 constant competitive algorithm for the online bipartite matching
 problem considered in Section \ref{sec:bipartite}, of which the
 transversal matroid secretary problem is a special case. In
 particular, they gave a $1/8$-approximation. Kesselheim et
 al. \cite{KRTV2013} provided a $1/e$-competitive algorithm, which is
 best possible, as discussed above.

For the graphic matroid secretary problem, Babaioff et al. \cite{BIKK2018} provide a deterministic $1/16$-competitive algorithm. This was improved to a $1/(3e)$-competitive algorithm by Babaioff et al. \cite{BDGIT2009}; a $1/(2e)$-competitive algorithm by Korula and P\'al \cite{KP2009}; and a $1/4$-competitive algorithm by Soto et al. \cite{STV2018}, which is currently the best algorithm. 
The algorithm from \cite{BIKK2018} is deterministic, whereas the other three are randomized.
All algorithms run in polynomial time.

\paragraph{Other models, extensions and variations.} 
There is a vast literature on online selection algorithms for problems similar to the (matroid) secretary problem. 
Here we discuss some recent directions and other models
  incorporating some form of prior information. 

The most important assumption in the secretary model that we consider is the fact that elements arrive in a uniformly random order. If the elements arrive in an adversarial order, there is not much one can achieve: There is a trivial randomized algorithm that selects every element with probability $1/n$, yielding a $1/n$-competitive algorithm; deterministically no finite competitive algorithm is possible with a guarantee independent of the values of the elements. There has been a recent interest in studying \emph{intermediate} arrival models that are not completely adversarial, 
nor uniformly random. Kesselheim, Kleinberg and Niazadeh \cite{KKN2015} study non-uniform arrival orderings under which (asymptotically) one can still obtain a $1/e$-competitive algorithm for the secretary problem. Bradac et al. \cite{BGSZ2020} consider the so-called Byzantine secretary model in which some elements arrive uniformly at random, but where an adversary controls a set of elements that can be inserted in the ordering of the uniform elements in an adversarial manner. See also the very recent work of Garg et al. \cite{GKRS2020} for a conceptually similar model.

In a slightly different setting, Kaplan et al.~\cite{KaplanNR20}
consider a secretary problem with the assumption that the algorithm
has access to a random sample of the adversarial distribution ahead of
time. For this setting they provide an algorithm with almost tight
competitive-ratio for small sample-sizes. 

Furthermore, there is also a vast literature on so-called \emph{prophet inequalities}. 
In the basic model, the elements arrive in an adversarial order, but there is a prior distributional information given for the values of the elements $\{1,\dots,n\}$. That is, one is given probability distributions $X_1,\dots,X_n$ from which the values of the elements are drawn. 
Upon arrival of an element $e$, its value drawn according to $X_e$ is revealed and an irrevocable decision is made whether to select this element or not.
Note that the available distributional information can
be used to decide on whether to select an element. 
The goal is to maximize the expected value, taken over all prior distributions, of the selected element.
For surveys on recent developments, refer to \cite{L2017,CFHOV2019}. Here we discuss some classical results and recent related works. 
Krengel, Sucheston and Garling \cite{KS1978} show that there is an optimal $1/2$-competitive algorithm for this problem. Kleinberg and Weinberg \cite{KW2019} gave a significant generalization of this result to matroid prophet inequalities, where multiple elements can be selected subject to a matroid feasibility constraint (an analogue of the matroid secretary problem). 
There is also a growing interest in the prophet secretary problem \cite{Ehsani2018}, in which the elements arrive uniformly random (as in the secretary problem); see also \cite{CFHOV2019}. 

Recently, settings with more limited prior information gained a lot of interest. These works address the quite strong assumption of knowing all element-wise prior distributions.
Azar, Kleinberg and Weinberg \cite{AKW2014} study the setting in which one has only access to one sample from every distribution, as opposed to the whole distribution; see also \cite{W2018}. Correa et al. \cite{CDFS2019} study this problem under the assumption that all elements are identically distributed. 
Recently, an extension of this setting was considered by Correa et al. \cite{CCES2020}.
Furthermore, D\"utting and Kesselheim \cite{DK2019} consider prophet inequalities with inaccurate prior distributions $\tilde{X}_1,\dots,\tilde{X}_n$ (while the true distributions $X_1,\dots,X_n$ are unknown). They study to what extent the existing algorithms are robust against inaccurate prior distributions.\\

\noindent Although our setting also assumes additional information about the input instance, there are major differences.
Mainly, we are interested in including a minimal amount of 
predictive information about an optimal (offline) solution, which yields a quantitative improvement in the case where the prediction 
	is sufficiently accurate. This is a much weaker assumption than having 
	a priori all element-wise (possibly inaccurate) probability distributions.
Furthermore, our setting does not assume that the predictive information
necessarily comes from a distribution 
(which is then used to measure the expected performance of an  algorithm), 
but can be obtained in a more general fashion from historical data
(using, e.g., statistical or machine learning techniques).
Finally, and in contrast to other settings, the information received in our setting can be inaccurate (and this is non-trivial do deal with). 

\paragraph{Other problems with black-box machine learned advice}
Although online algorithms with machine learned advice are a relatively
new area, there has already been a number of  interesting
results.   
 We note that most of the following results are analyzed by means of
\emph{consistency} (competitive-ratio in the case of perfect
predictions) and \emph{robustness} (worst-case competitive-ratio
regardless of prediction quality), but the precise
formal definitions of consistency and robustness  slightly
differ in each paper \cite{LykourisV18,PurohitSK18}. \footnote{The
  term ``robustness'' has also been used in connection with secretary
  problems (see for example~\cite{BGSZ2020}), but in a totally different sense.} Our results can also be interpreted within this framework, but for the sake of completeness we give the competitive ratios as a function of the prediction error.
Purohit et al.~\cite{PurohitSK18}, considered the \emph{ski rental
problem} and the \emph{non-clairvoyant scheduling problem}. For both
problems they gave algorithms that are both consistent and robust,
and with a flavor similar to ours, the robustness and consistency of their algorithms are given as a function of some hyper parameter which
has to be chosen by the algorithm in advance. 
For the ski rental problem, in particular, 
Gollapudi et al.~\cite{GollapudiP19} considered the setting with multiple predictors,
and they provided and evaluated experimentally tight algorithms.

Lykouris and Vassilvitskii~\cite{LykourisV18} studied the
\emph{caching problem} (also known in the literature as \emph{paging}), and
were able to adapt the classical Marker algorithm~\cite{FiatRR94} 
to obtain a trade-off between robustness and consistency, for this problem.
Rohatgi~\cite{Rohatgi20} subsequently gave an algorithm 
	whose competitive ratio has an improved dependence on the prediction errors.

Further results in online algorithms with machine learned advice include
the work by Lattanzi et al.~\cite{LattanziLMV20} who studied 
the \emph{restricted assignment scheduling problem} with predictions 
on some dual variables associated to the machines, and the work by
Mitzenmacher~\cite{Mitzenmacher20} who considered a different 
\emph{scheduling/queuing problem}. They introduced a novel
quality measure for evaluating algorithms, called the \emph{price of misprediction}.

Finally, Mahdian et al.~\cite{MahdianNS12} studied problems 
where it is assumed that there exists an
optimistic algorithm (which could in some way be interpreted as a
prediction), and designed a meta-algorithm that interpolates between
a worst-case algorithm and the optimistic one.
They considered several problems, 
including the \emph{allocation of online advertisement space},
and for each gave an algorithm whose competitive ratio 
is also an interpolation between the competitive ratios 
of its corresponding optimistic and worst-case algorithms.
However, the performance guarantee is not given as 
a function of the ``prediction'' error, but rather only as 
a function of the respective ratios and the interpolation parameter.

%% file: preliminaries.tex
\section{Preliminaries}
In this section we formally define the online algorithms of interest, provide the necessary graph theoretical notation, and define the so-called Lambert $W$-function that will be used in Section \ref{sec:secretary}.

\subsection{Online algorithms with uniformly random arrivals}
We briefly sketch some relevant definitions for the online problems that we consider in this work. We consider online selection problems in which the goal is to select
the ``best feasible'' subset 
out of a finite set of
objects $O$ with size $|O|=n$, that arrive online in a uniformly random order. More formally, the $n$ objects are
revealed to the algorithm one object per round. In each round $i$, and upon revelation of
the current object $o_i\in O$, the online selection algorithm has to irrevocably
select an \emph{outcome} $z_i$ out of a set of possible outcomes
$Z(o_i)$ (which may depend on $o_1,o_2,\dots o_i$ as well as
$z_1,z_2,\dots z_{i-1}$.) Each outcome $z_i$ is associated with a \emph{value}
$v_i(z_i)$, and all values $v_i$ become known to the algorithm with the
arrival of $o_i$. The goal is to maximize the total value $T= \sum_i v_i(z_i)$.

The cost of an algorithm $\mathcal{A}$ selecting outcomes
$z_1,z_2\dots z_n$ on $\sigma$ is defined as $T(\mathcal{A}(\sigma)) =
\sum_i v_i(z_i)$. Such an algorithm $\mathcal{A}$ is $\gamma$-competitive if 
$
\mathbb{E}(T(\mathcal{A}(\sigma))) \geq \gamma \cdot \text{OPT}(\sigma),
$
for $0 < \gamma \leq 1$, where \text{OPT}$(\sigma)$ is the objective
value of an offline optimal solution, i.e., one that is aware of the whole input sequence $\sigma$
in advance.  The
expectation is taken over the randomness in $\sigma$ (and the internal
randomness of $\mathcal{A}$ in case of a randomized
algorithm). Alternatively, we say that $\mathcal{A}$ is a
$\gamma$-approximation.

\subsection{Graph theoretical notation}
An undirected graph $G = (V,E)$ is defined by a set of nodes $V$ and
set of edges $E \subseteq \{ \{u,v\} : u,v \in V, u\neq v\}$. A bipartite graph $G = (L \cup R, E)$ is given by two sets of nodes $L$ and $R$, and $E \subseteq \{\{\ell,r\} : \ell \in L, r \in R\}$. In the bipartite case we sometimes write $(\ell,r)$ to indicate that $\ell \in L$ and $r \in R$ (we also use this notation for directed arcs in general directed graphs). For a set of nodes $W$, we use $G[W]$ to denote the induced (bipartite) subgraph on the nodes in $W$.

A function $w : E \rightarrow \R_{\geq 0}$ is called a weight function on the edges in $E$; we sometimes write $w(\ell,r)$ in order to denote $w(\{u,v\})$ for $\{u,v\} \in E$.  A matching $M \subseteq E$ is a subset of edges so that every node is adjacent to at most one edge in $M$. For a set of nodes $W$, we write $W[M]$ to denote the nodes in $W$ that are adjacent to an edge in $M$. Such nodes are said to be matched. If $G$ is undirected, we say that $M$ is perfect if every node in $V$ is adjacent to precisely one edge in $M$. If $G$ is bipartite, we say that $M$ is perfect w.r.t $L$ if every $\ell \in L$ is adjacent to one edge in $M$, and perfect w.r.t $R$ if every $r \in R$ is adjacent to some edge in $M$.

\begin{remark}
When $G$ is bipartite, we will assume that for every subset $S \subseteq L$, there is a perfect matching w.r.t $S$ in $G[S \cup R]$. This can be done without loss of generality by adding for every $\ell \in L$ a node $r^\ell$ to the set $R$, and adding the edge $\{\ell,r^\ell\}$ to $E$. Moreover, given a weight function on $E$ we extend it to a weight function on $E'$ by giving all the new edges weight zero.
\end{remark}

\subsection{Lambert $W$-function}
The Lambert $W$-function is the inverse relation of the function $f(w) = we^w$. Here, we consider this function over the real numbers, i.e., the case $f : \R \rightarrow \R$. Consider the equation
$
ye^y = x. 
$
For $-1/e \leq x < 0$, this equation has two solutions denoted by $y = W_{-1}(x)$ and $y = W_{0}(x)$, where $W_{-1}(x) \leq W_0(x)$ with equality if and only if $x = -1/e$.

%% file: secretary.tex
\section{Secretary problem}\label{sec:secretary}
In the secretary problem there is a set $\{1,\dots,n\}$ of secretaries, each with a value $v_i \geq 0$ for $i \in \{1,\dots,n\}$, that arrive in a \emph{uniformly random order}. Whenever a secretary arrives, we have to irrevocably decide whether we want to hire that person. If we decide to hire a secretary, we automatically reject all subsequent candidates. The goal is to select the secretary with the highest value. We assume without loss of generality that all values $v_i$ are distinct. This can be done by introducing a suitable tie-breaking rule if necessary.

There are  two versions of the secretary problem. In the \emph{classical secretary problem}, the goal is to maximize the probability with which the best secretary is chosen. We consider a slightly different version, where the goal is to maximize the expected value of the chosen secretary. We refer to this as the \emph{value-maximization secretary problem}.\footnote{Any $\alpha$-approximation for the classical secretary problem yields an $\alpha$-approximation for the value-maximization variant.} In the remainder of this  work, the term `secretary problem' will always refer to the value-maximization secretary problem, unless stated otherwise.
The (optimal) solution \cite{L1961,D1962} to both variants of the secretary problem is to first observe a fraction of $n/e$ secretaries. After that, the first secretary with a value higher than the best value seen in the first fraction is selected. This yields a $1/e$-approximation for both versions. 

The machine learned advice that we consider in this section is a prediction $p^*$ for the maximum value $\text{OPT} = \max_i v_i$ among all secretaries.\footnote{Such a prediction does not seem to have any relevance in case the goal is to maximize the probability with which the best secretary is chosen. Intuitively, for every instance $(v_1,\dots,v_n)$ there is a `more or less isomorphic' instance $(v_1',\dots,v_n')$ for which $v_i' < v_j'$ if and only if $v_i < v_j$ for all $i,j$, and for which, all values are close to each other and also to the prediction $p^*$, for which we assume that $p^* < \min_i v_i'$. Any algorithm that uses only pairwise comparisons between the values $v_i'$ can, intuitively, not benefit from the prediction $p^*$. Of course, for the value-maximization variant, choosing any $i$ will be close to optimal in this case if the prediction error is small.} Note that we do not predict which secretary has the highest value. We define the \emph{prediction error} as
$$
\eta = |p^* - \text{OPT}|.
$$
We emphasize that this parameter is not known to the algorithm, but is only used to analyse the algorithm's performance guarantee.

If we would \emph{a priori} know that the prediction error is small, then it seems reasonable to pick the first element that has a value `close enough' to the predicted optimal value. One small issue that arises here is that we do not know whether the predicted value is smaller or larger than the actual optimal value. In the latter case, the entire Phase II of the algorithm would be rendered useless, even if the prediction error was arbitrarily small. In order to circumvent this issue, one can first lower the predicted value $p^*$ slightly by some $\lambda > 0$ and then select the first element that is greater or equal than the threshold $p^* - \lambda$.\footnote{Alternatively, one could define an interval around $p^*$, but given that we get a prediction for the maximum value, this does not make a big difference.} Roughly speaking, the parameter $\lambda$ can be interpreted as a guess for the prediction error $\eta$.

\IncMargin{1em}
\begin{algorithm}
\SetKwData{Left}{left}\SetKwData{This}{this}\SetKwData{Up}{up}
\SetKwFunction{Union}{Union}\SetKwFunction{FindCompress}{FindCompress}
\SetKwInOut{Input}{Input}\SetKwInOut{Output}{Output}
\Input{Prediction $p^*$ for (unknown) value $\max_i v_i$; confidence parameter $0 \leq \lambda \leq p^*$ and $c \geq 1$.}
\Output{Element $a$.}\medskip
 
Set $v' = 0$.\\
\textbf{Phase I:} \\
\For{$i = 1,\dots,\lfloor\exp\{W_{-1}(-1/(ce))\} \cdot n\rfloor$}{
Set $v' = \max\{v',v_i\}$}
Set $t = \max\{v',p^* - \lambda\}$.\\
\textbf{Phase II:} \\
\For{$i = \lfloor\exp\{W_{-1}(-1/(ce))\} \cdot n\rfloor + 1, \dots, \lfloor \exp\{W_{0}(-1/(ce))\} \cdot n\rfloor$}{
\If{$v_i >  t$}{
Select element $a_i$ and STOP.}}
Set $t = \max\{v_{j}\ :\ j\in\{1,\dots,\lfloor\exp(W_{0}(-1/(ce)))\cdot n\rfloor\}\}$. \\
\textbf{Phase III:}\\
\For{$i = \lfloor\exp\{W_{0}(-1/(ce))\} \cdot n\rfloor + 1, \dots, n$}{
\If{$v_i >  t$}{
Select element $a_i$ and STOP.}}
\caption{Value-maximization secretary algorithm}
\label{alg:classical}
\end{algorithm}
\DecMargin{1em}

Algorithm \ref{alg:classical} incorporates the idea sketched in the previous paragraph (while at the same time guaranteeing a constant competitive algorithm if the prediction error is large).  It is a generalization of the well-known optimal algorithm for both the classical and value-maximization secretary problem \cite{L1961,D1962}. The input parameter $\lambda$  is a confidence parameter in the prediction $p^*$ that allows us to interpolate between the following two extreme cases:
\begin{enumerate}[i)]
\item If we have very low confidence in the prediction, we choose $\lambda$ close to $p^*$;
\item If we have high confidence in the prediction, we choose $\lambda$ close to $0$.
\end{enumerate}
In the first case, we essentially get back the classical solution \cite{L1961,D1962}. Otherwise, when the confidence in the prediction is high, we get a competitive ratio better than $1/e$ in case the prediction error $\eta$ is in fact small, in particular, smaller than $\lambda$. If our confidence in the prediction turned out to be wrong, when $\lambda$ is larger than the prediction error, we still obtain a $1/(ce)$-competitive algorithm. The parameter $c$ models what factor we are willing to lose in the worst-case when the prediction is  poor (but the confidence in the prediction is high). In Theorem \ref{thm:classical} below, we analyze Algorithm \ref{alg:classical}.

\begin{theorem}\label{thm:classical}
	For any $\lambda \ge 0$ and $c > 1$, there is a deterministic algorithm for the (value-maximization) secretary problem that is asymptotically $g_{c,\lambda}(\eta)$-competitive in expectation, where
	\[
	g_{c,\lambda}(\eta) = \left\{
	\begin{array}{ll}
	\max\left\{\frac{1}{ce},\left[f(c)\left(\max\left\{1 - \frac{\lambda + \eta}{OPT},0\right\}\right)\right]\right\} & \text{ if } 0 \leq \eta < \lambda\\
	\frac{1}{ce} & \text{ if } \eta \geq \lambda
	\end{array}\right\},
	\]
	and the function $f(c)$ is given in terms of the two branches 
	$W_0$ and $W_{-1}$ of the Lambert $W$-function and reads
	\[ f(c)=\exp\{W_0(-1/(ce))\} - \exp\{W_{-1}(-1/(ce))\}. \]
\end{theorem}

We note that $\lambda$ and $c$ are independent
parameters that provide the most general description of the
competitive ratio. Here $\lambda$ is our confidence
of the predictions and $c$ describes how much we are willing to lose
in the worst case. Although these parameters can be set independently,
some combinations of them are not very sensible, as one might
not get an improved performance guarantee, even when the prediction
error is small (for instance, if $c = 1$, i.e., we are not willing to
lose anything in the worst case, then it is not helpful to consider
the prediction at all).
To illustrate the influence of these parameters
	on the competitive ratio, in Figure~\ref{fig:secretary},
	we plot various combinations of the input parameters 
	$c, \lambda$ and $p^*$ of Algorithm~\ref{alg:classical} (in Section~\ref{sec:secretary}), 
	assuming that $\eta = 0$. In this case $p^* = \text{OPT}$ and
	the competitive ratio simplifies to
	$$
	g_{c,\lambda}(0) = \max\left\{\frac{1}{ce},\, f(c)\cdot\max\left\{1 - \frac{\lambda}{p^*},0\right\}\right\}.
	$$ 
	We therefore choose the axes of Figure~\ref{fig:secretary} to be $\lambda/p^*$ and $c$.

\begin{figure}[h]
	\centering
	\includegraphics[scale=0.6	]{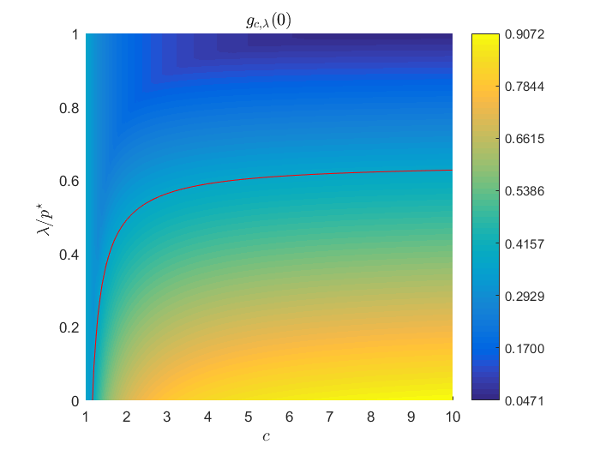}
	\caption{The red curve shows the optimal competitive ratio 
		without predictions, i.e.,  $g_{c,\lambda}(0)=1/e$.
		Our algorithm achieves an improved competitive ratio $g_{c,\lambda}(0)>1/e$
		in the area below this curve, and a worse competitive ratio
		$g_{c,\lambda}(0)<1/e$ in the area above it.}
	\label{fig:secretary}
\end{figure}

Furthermore, as one does not know the
prediction error $\eta$, there is no way in choosing these parameters
optimally, since different prediction errors require different
settings of $\lambda$ and $c$.

To get an impression of the statement in Theorem \ref{thm:classical},
if we have, for example, $\eta + \lambda =\frac{1}{10}\text{OPT}$, then
we start improving over $1/e$ for $c \geq 1.185$. Moreover, if one
believes that the prediction error is low, one should set $c$ very
high (hence approaching a $1$-competitive algorithm in case the
predictions are close to perfect).

\begin{remark}
Note, that the bound obtained in
Theorem~\ref{thm:classical} has a discontinuity at $\eta =
\lambda$. This can be easily smoothed out by selecting $\lambda$
according to some distribution, which now represents
our confidence in the prediction $p^*$. The competitive ratio will start
to drop earlier in this case, and will continuously reach
$1/(ce)$. Furthermore, for $\eta=\lambda=0$ this
  bound is tight for any fixed $c$.  We illustrate
how the competitive ratio changes as a function of $\eta$ in Appendix~\ref{app:randAlgo}.

Moreover, in Appendix~\ref{app:CompNRndAlgo}, we show that our deterministic Algorithm~\ref{alg:classical}
	can achieve a better competitive ratio than its corresponding naive randomization.
\end{remark}

\begin{proof}[Proof of Theorem~\ref{thm:classical}]
By carefully looking into the analysis of the classical secretary
problem, see, e.g., \cite{D1962,L1961}, it becomes clear that although
sampling an $1/e$-fraction of the items is the optimal trade-off for
the classical algorithm and results in a competitive ratio of $1/e$,
one could obtain a $1/(ce)$-competitive ratio (for $c>1$) in two ways: by sampling either less, or more items, more specifically an $\exp\{W_{-1}(-1/(ce))\}$ or $\exp\{W_{0}(-1/(ce))\}$ fraction of the items respectively. These quantities arise as the two solutions of the equation
$$
-x\ln x = \frac{1}{ce}.
$$
We next provide two lower bounds on the competitive ratio. 

 First of all, we prove that in the worst-case we are always $1/(ce)$-competitive. We consider two cases.

\textit{Case 1: $p^* - \lambda > OPT$.} Then we never pick an element in Phase II, which means that the algorithm is equivalent to the algorithm that observes a fraction $\exp\{W_{0}(-1/(ce))\}$ of all elements and then chooses the first element better than what has been seen before, which we know is $1/(ce)$-competitive.

\textit{Case 2. $p^* - \lambda \leq OPT$.} Consider a fixed arrival order and suppose that, for this permutation, we select OPT in the algorithm that first observes a fraction $\exp\{W_{-1}(-1/(ce))\}$ of all elements and then selects the first element better than what has been seen before (which we know is $1/(ce)$-competitive). It should be clear that our algorithm also chooses OPT in this case. As the analysis in \cite{D1962,L1961} relies on analyzing the probability with which we pick OPT, it follows that our algorithm is also $1/(ce)$-competitive in this case.

 The second bound on the competitive ratio applies to cases in which the prediction error is small. In particular, suppose that $0 \leq \eta < \lambda$.

\textit{Case 1: $p^* > OPT$.} We know that $p^* - \lambda < OPT$, as $\eta < \lambda$. Therefore, if $OPT$ appears in Phase II, and we have not picked anything so far,  we will pick $OPT$. Since $OPT$ appears in Phase II with probability $f(c)$, we in particular pick some element in Phase II with value at least $OPT - \lambda$ with probability $f(c)$ (note that this element does not have to be $OPT$ necessarily). 

\textit{Case 2: $p^* \leq OPT$.} In this case, using similar reasoning as in Case 1, with probability $f(c)$ we will pick some element with value at least $OPT - \lambda - \eta$. To see this, note that in the worst case we would have $p^* = OPT - \eta$, and we could select an element with value $p^* - \lambda$, which means that the value of the selected item is $OPT - \lambda - \eta$.

This means that, in any case, with probability at least $f(c)$, we will pick some element  in Phase II with value at least
$
\min\{OPT - \lambda, OPT- \lambda - \eta\} = OPT - \lambda - \eta 
$
if $\eta < \lambda$. That is, if $0 \leq \eta < \lambda$, and if we assume that $\text{OPT} - \lambda - \eta \geq 0$, we are guaranteed to be 
$
f(c)\left(1 - (\lambda + \eta)/\text{OPT}\right)\text{-competitive}.
$
\end{proof}

%% file: matching.tex
\section{Online bipartite matching with random arrivals}\label{sec:bipartite}
In this section we consider a generalization of the value-maximization
secretary problem discussed in Section \ref{sec:secretary}. We study
an online bipartite matching problem on a graph $G = (L \cup R,E)$
with edge weight function $w : E \rightarrow \R_{\geq 0}$. The
vertices $\ell \in L$ arrive in a uniformly random order. Whenever a
vertex $\ell \in L$ arrives, it reveals its neighbors $r \in R$ and
what the corresponding edge weights $w(\ell,r)$  are. We then have the option to add an edge of the form $(\ell,r)$, provided $r$ has not been matched in an earlier step. The goal is to select a set of edges, i.e., a matching, with maximum weight. 

We assume that we are given, for all offline nodes $r \in R$, a
prediction $p_r^*$ for the value of the edge weight adjacent to $r$ in some fixed optimal offline matching (which is zero if $r$ is predicted not to be matched in this offline matching). That is, we predict that there exists some fixed optimal offline matching in which $r$ is adjacent to an edge of weight $p_r^*$ without predicting which particular edge this is. 
Note that the predictions $p = (p_1^*,\dots,p_r^*)$ implicitly  provide a prediction for $\text{OPT}$, namely $\sum_{r \in R} p_r^*$.

It turns out that this type of predictions closely corresponds to the so-called \emph{online vertex-weighted bipartite matching problem} where every offline node is given a weight $w_r$, and the goal is to select a matching with maximum weight, which is the sum of all weights $w_r$ for which the corresponding $r$ is matched in the online algorithm. This problem has both been studied under adversarial arrivals \cite{KVV1990,AGKA2011} and uniformly random arrivals \cite{MY2011,HTWZ2019}. In case the predictions are perfect, then, in order to find a matching with the corresponding predicted values, we just ignore all edges $w(\ell,r)$ that do not match the value $p_r^*$. This brings us in a special case of the online vertex-weighted bipartite matching problem.

The prediction error in this section will be defined
  as the maximum error over all predicted values and the minimum over
  all optimal matchings in $G$. We use $\mathcal{M}(G,w)$ to denote the set of all optimal matchings in $G$ with respect to the weight function $w$, and then define
$$
\eta = \min_{M \in \mathcal{M}(G,r)}\max_{r \in R} |p_r^* - w(M_r)|.
$$
Here, we use $w(M_r)$ to denote the value of the edge adjacent to $r \in R$ in a given optimal solution with objective value $\text{OPT} = \sum_r w(M_r)$.

In the next sections we will present deterministic and randomized algorithms, inspired by algorithms for the online vertex-weighted bipartite matching problem, that can be combined with the algorithm in \cite{KRTV2013} in order to obtain algorithms that incorporate the predictions and have the desired properties. We start with a deterministic algorithm, which is the main result of this section.

\subsection{Deterministic algorithm}
We first give a simple deterministic greedy algorithm  that provides a
$1/2$-approximation in the case when the predictions are perfect
(which is true even for an adversarial arrival order of the nodes). It
is very similar to a greedy algorithm given by Aggarwal et
al. \cite{AGKA2011}. Although we do not emphasize it in the
description, this algorithm can be run in an online fashion. 

\IncMargin{1em}
\begin{algorithm}
\SetKwData{Left}{left}\SetKwData{This}{this}\SetKwData{Up}{up}
\SetKwFunction{Union}{Union}\SetKwFunction{FindCompress}{FindCompress}
\SetKwInOut{Input}{Input}\SetKwInOut{Output}{Output}
\Input{Thresholds $t = (t_1,\dots,t_{|R|})$ for offline nodes $r \in R$; ordered list $(v_1,\dots,v_\ell) \subseteq L$.}
\Output{Matching $M$}\medskip

Set $M = \emptyset$.\\
\For{$i = 1,\dots,\ell$}{
Set $r^i = \text{argmax}_r\{w(v_i,r) : r \in \mathcal{N}(v_i), w(v_i,r) \geq t_r \text{ and } r \notin R[M]\}$.\\
\If{$r^i \neq \emptyset$}{
Set $M = M \cup \{v_i,r^i\}$.
}}
\caption{Threshold greedy algorithm}
\label{alg:greedy}
\end{algorithm}
\DecMargin{1em}
Provided that there exists an offline matching in which every $r \in
R$ is adjacent to some edge with weight at least $t_r$, it can be
shown, using the same arguments as given in \cite{AGKA2011}, that the
threshold greedy algorithm yields a matching with weight at least
$\frac{1}{2}\sum_{r} t_r$. We present the details in the proof of Theorem \ref{thm:bipartite} later on.

It is also well-known that, even for uniformly random arrival order and
unit edge weights, one cannot obtain anything better than a $1/2$-approximation with a deterministic algorithm \cite{AGKA2011}.\footnote{The example given there is for adversarial arrival order, but also applies to uniformly random arrival order.} This also means that, with our choice of predictions, we cannot do better than a $1/2$-approximation in the ideal case in which the predictions are perfect. Therefore, our aim is to give an algorithm that includes the predictions in such a way that, if the predictions are good (and we have high confidence in them), we should approach a $1/2$-approximation, whereas if the predictions turn out to be poor, we are allowed to lose at most a constant factor w.r.t. the $1/e$-approximation in \cite{KRTV2013}.

Algorithm \ref{alg:bipartite} is a deterministic algorithm satisfying
these properties, that, similar to Algorithm \ref{alg:classical},
consists of three phases. The first two phases correspond to the
algorithm of Kesselheim et al. \cite{KRTV2013}. In the third phase, we
then run the threshold greedy algorithm as described in Algorithm
\ref{alg:greedy}. Roughly speaking, we need to keep two things in
mind. First of all, we should not match up too many offline nodes in
the second phase, as this would block the possibility of selecting a
good solution in the third phase in case the predictions are good. On
the other hand, we also do not want to match too few offline nodes in the second phase, otherwise we are no longer guaranteed to be constant-competitive in case the predictions turn out to be poor. The analysis of Algorithm \ref{alg:bipartite} given in Theorem \ref{thm:bipartite} shows that it is possible to achieve both these properties. 

For the sake of simplicity, in both the description of Algorithm \ref{alg:bipartite} and its analysis in Theorem \ref{thm:bipartite}, we use a common $\lambda$ to lower the predicted values (as we did in Section \ref{sec:secretary} for the secretary problem). Alternatively, one could use a resource-specific value $\lambda_r$ for this as well.

\IncMargin{1em}
\begin{algorithm}
\SetKwData{Left}{left}\SetKwData{This}{this}\SetKwData{Up}{up}
\SetKwFunction{Union}{Union}\SetKwFunction{FindCompress}{FindCompress}
\SetKwInOut{Input}{Input}\SetKwInOut{Output}{Output}
\Input{Predictions $p^*= (p_1^*,\dots,p_{|R|}^*)$, confidence parameter $0 \leq \lambda \leq \min_r p_r^*$, and $c > d \geq 1$.}
\Output{Matching $M$.}\medskip

\textbf{Phase I:} \hfill \textit{/*Algorithm from \cite{KRTV2013}} \\
\For{$i = 1,\dots,\lfloor n/c \rfloor$}{
Observe arrival of node $\ell_i$, and store all the edges adjacent to it.}
Let $L' = \{\ell_1,\dots,\ell_{\lfloor n/c \rfloor}\}$ and $M = \emptyset$.\\
\textbf{Phase II:} \\
\For{$i = \lfloor n/c \rfloor + 1, \dots, \lfloor n/d \rfloor$}{
Set $L' = L' \cup \ell_i$.\\
Set $M^{i} = \text{optimal matching on } G[L' \cup R]$.\\
Let $e^{i} = (\ell_i,r)$ be the edge assigned to $\ell_i$ in $M^{i}$.\\
\If{$M \cup e^i$ is a matching}{
Set $M = M \cup \{e^i\}$.}}
\textbf{Phase III:} \hfill \textit{/*Threshold greedy algorithm}\\
\For{$i =  \lfloor n/d \rfloor + 1,\dots,n$}{
Set $r^i = \text{argmax}_r\{w(v_i,r) : r \in \mathcal{N}(v_i), w(v_i,r) \geq p_r^* - \lambda \text{ and } r \notin R[M]\}$\\
\If{$r^i \neq \emptyset$}{
Set $M = M \cup \{\ell_i,r\}$.
}}
\caption{Online bipartite matching algorithm with predictions}
\label{alg:bipartite}
\end{algorithm}
\DecMargin{1em}
\begin{theorem}\label{thm:bipartite}
	For any $\lambda \ge 0$ and $c>d\ge 1$, there is a deterministic algorithm for the online bipartite matching problem with uniformly random arrivals that is asymptotically $g_{c,d,\lambda}(\eta)$-competitive in expectation, where
	\[
	g_{c,d,\lambda}(\eta) = \left\{
	\begin{array}{ll}
	\max\left\{\frac{1}{c}\ln(\frac{c}{d}),\left[\frac{d-1}{2c} \left(\max\left\{1 - \frac{(\lambda + \eta)|\psi|}{\text{OPT}},0\right\}\right)\right]\right\} & \text{ if } 0 \leq \eta < \lambda,\\
	\frac{1}{c}\ln(\frac{c}{d}) & \text{ if } \eta \geq \lambda.
	\end{array}\right\},
	\]
	and $|\psi|$ is the cardinality of an optimal (offline) matching $\psi$ of the instance.
\end{theorem}
\noindent If $\lambda$ is small, we roughly obtain a bound of $(d-1)/2c$ in case $\eta$ is small as well, and a bound of $\ln(c/d)/c$ if $\eta$ is large. Moreover, when $c/d \rightarrow 1$, we approach a bound of $1/2$ in case $\eta$ is small, whereas the worst-case guarantee $\ln(c/d)/c$ for large $\eta$ then increases.
\begin{proof}[Proof of Theorem \ref{thm:bipartite}]
We provide two lower bounds on the expected value of the matching $M$ output by Algorithm \ref{alg:bipartite}. 

First of all, the analysis of the algorithm of Kesselheim et al. \cite{KRTV2013} can be generalized to the setting we consider in the first and second phase. In particular, their algorithm then yields a
$$
\left(\frac{1}{c} - \frac{1}{n}\right) \ln\left(\frac{c}{d}\right)\text{-competitive approximation}.
$$
For completeness, we present a proof of this statement in Appendix \ref{app:kesselheim}.\\

\noindent The second bound we prove on the expected value of the matching $M$ is based on the threshold greedy algorithm we use in the third phase. Let $\psi \in \mathcal{M}(G,w)$ be an optimal (offline) matching, with objective value $\text{OPT}$, and suppose that
$$
\eta = \max_{r \in R} |\psi_r - p_r^*| < \lambda. 
$$
The proof of the algorithm in \cite{KRTV2013} analyzes the expected
value of the online vertices in $L$. Here we take a different approach
and study the expected value of the edge weights adjacent to the nodes in $r \in R$.
Fix some $r \in R$ and consider the edge $(\ell,r)$ that is matched to $r$ in the optimal offline matching $\psi$ (if any). 

Let $X_r$ be a random variable denoting the value of node $r \in R$ in the online matching $M$ chosen by Algorithm \ref{alg:bipartite}. Let $Y_{\ell}$ be a random variable that denotes the value of node $\ell \in L[\psi]$ in the online matching $M$. It is not hard to see that
\begin{equation}\label{eq:expected}
\mathbb{E}(M) = \sum_{r \in R} \mathbb{E}(X_r) \ \ \text{ and } \ \ \mathbb{E}(M) \geq \sum_{\ell \in L[\psi]} \mathbb{E}(Y_{\ell}).
\end{equation}
For the inequality, note that for any fixed permutation the value of
the obtained matching is always larger or equal to the sum of the values that were matched to the nodes $\ell \in L[\psi]$.

Now, consider a fixed $\ell$ and $r$ for which the edge $(\ell,r)$ is contained in $\psi$. We will lower bound the expectation $\mathbb{E}(X_r + Y_{\ell})$ based on the expected value these nodes would receive, roughly speaking, if they get matched in the third phase. Therefore, suppose for now that $r$ did not get matched in the second phase, and that $\ell$ appears in the part of the uniformly random permutation considered in the third phase. We will later lower bound the probability with which the events occur. By definition of the greedy threshold algorithm, we know that at the end of Phase III either node $r \in R$ is matched, or otherwise at least node $\ell$ is matched to some other $r' \in R$ for which 
$$
w(\ell,r') \geq w(\ell,r) \geq p_r^* - \eta\geq p_r^* - \lambda. 
$$
To see this, note that when the vertex $\ell$ arrived, there was the option to match it to $r$, as this node is assumed not to have been matched in the second phase. So, there are the following three cases: either $\ell$ got matched to $r$, or it got matched to some other $r'$ for which $w(\ell,r') \geq w(\ell,r)\geq p_r^* - \lambda$ or $r$ was matched earlier during the third phase  to some other $\ell'$ for which $w(\ell',r) \geq p_r^* - \lambda$. 

Looking closely at the analysis of Kesselheim et al. \cite{KRTV2013}, see Appendix~\ref{sub:PhaseIII}, it follows that the probability that a fixed node $r$ did not get matched in the second phase satisfies 
$$
\text{P}(r \text{ was not matched in Phase II}) \geq \frac{d}{c} - o(1).
$$
This lower bound is true independent of whether or not $\ell$ appeared in Phase III, or the first two phases (see Appendix \ref{app:kesselheim}). This former event holds with probability $(1 - 1/d)$. Therefore,
$$
\text{P}(r \text{ was not matched in Phase II \textit{and} } \ell \text{ arrives in Phase III}) \geq \left(\frac{d}{c} - o(1)\right)\left(1 - \frac{1}{d}\right) = \frac{d-1}{c} - o(1).
$$
Furthermore, we have from an earlier argument that under this condition either $X_r\geq p_r^* - \lambda$ or $Y_\ell\geq p_r^* - \lambda$. This implies that
\begin{equation}\label{eq:r_ell}
\mathbb{E}(X_r + Y_{\ell}) \geq\left(\frac{d-1}{c} - o(1)\right) (p_r^* - \lambda),
\end{equation}
and combining this with \eqref{eq:expected}, we find
$$
2\cdot \mathbb{E}(M) \geq \left(\frac{d-1}{c} - o(1)\right) \sum_{r \in R[\psi]} (p_r^* - \lambda) \geq \left(\frac{d-1}{c} - o(1)\right) (\text{OPT} - (\lambda + \eta)|\psi|)
$$
assuming that $\text{OPT} - (\lambda + \eta)|\psi| \geq 0$.
Rewriting this gives
$$
\mathbb{E}(M) \geq \left(\frac{d-1}{2c} - o(1)\right) \left(1 - \frac{(\lambda + \eta)|\psi|}{\text{OPT}}\right) \cdot \text{OPT},
$$
which yields the desired bound.
\end{proof}

In order to get close to a $1/2$-approximation in case the predictions are (almost) perfect, we have to choose both $c$ and $d$ very large, as well as the ratio $c/d$ close to $1$ (but still constant). It is perhaps also interesting to note that Theorem \ref{thm:bipartite} does not seem to hold in case we interchange the second and third phase. In particular, if the predictions are too low, we most likely match up too many nodes in $r \in R$ already in the second phase (that now would execute the threshold greedy algorithm).

\subsection{Truthful mechanism for single-value unit-demand domains}
Algorithm \ref{alg:bipartite} can be turned into a truthful mechanism for \emph{single-value unit-demand domains}. Here, every node arriving online corresponds to an \emph{agent} $i$ that is interested in a subset $A_i$ of the \emph{items} in $R$ (i.e., her neighbors in $R$) between which she is indifferent. That is, agent $i$ has a value $v_i$ for all items in $A_i$, and zero for all other items. Whenever an agent arrives, she reports a value $v_i'$ for the items in $A_i$ that she is interested in (we assume that these sets are common knowledge). Based on the value $v_i'$, the mechanism decides if it wants to allocate one of the items in $A_i$ (not yet assigned to any agent) to $i$, and, if so, sets a price $\rho_i$ for this item that is charged to agent $i$. The goal is to choose an allocation that maximizes the social welfare (which is the total value of all assigned items). 

The goal of the agents is to maximize their utility which is their valuation minus the price that they are charged by the mechanism for the item assigned to them (if any).
We want to design a mechanism that incentivizes agents to report their true value $v_i$.

In Appendix \ref{app:truthful}, we provide formal definitions of all notions and show that Algorithm \ref{alg:bipartite} can be turned into a truthful mechanism for which its social welfare guarantee is $g_{\lambda,c,d}(\eta)$ as in Theorem \ref{thm:bipartite}. In particular, we do this by exploiting some of the freedom we have in Algorithm \ref{alg:bipartite}.

\subsection{Randomized algorithm}
If we allow randomization, we can give better approximation guarantees than the algorithm given in the previous section by using a convex combination of the algorithm of Kesselheim et al. \cite{KRTV2013}, and the randomized algorithm of Huang et al. \cite{HTWZ2019} for online vertex-weighted bipartite matching with uniformly random arrivals.

We give a simple, generic way to reduce an instance of online bipartite matching with predictions $p_r^*$ for $r \in R$  to an instance of online vertex-weighted bipartite matching with vertex weights (that applies in case the predictions are accurate). This reduction works under both a uniformly random arrival order of the vertices, as well as an adversarial arrival order. 

Suppose we are given an algorithm $\mathcal{A}$ for instances of the online vertex-weighted bipartite matching problem. Now, for the setting with predictions, fix some parameter $\lambda > 0$ up front. Whenever a vertex $\ell$ arrives online we only take into account edges $(\ell,r)$ with the property that $w(\ell,r) \in [p_r^* - \lambda,p_r^* + \lambda]$, and ignore all edges that do not satisfy this property. We then make the same decision for $\ell$ as we would do in algorithm $\mathcal{A}$ (only considering the edges that satisfy the interval constraint given above) based on assuming those edges have weight $p_r^*$. Then the matching that algorithm $\mathcal{A}$ generates has the property that its objective value is close (in terms of $\lambda$ and $\eta$) to the objective value under the original weights (and this remains to hold true in expectation in case of a uniformly random arrival order of the vertices).

The $0.6534$-competitive algorithm of Huang et al. \cite{HTWZ2019} is the currently best known randomized algorithm $\mathcal{A}$,  for online vertex-weighted bipartite matching with uniformly random arrivals, and can be used for our purposes. Detailed statements will be given in the journal version of this work.

%% file: graphic_matroid.tex
\section{Deterministic graphic matroid secretary algorithm}\label{sec:graphic}
In this section we will study the graphic matroid secretary
problem. Here, we are given a (connected) graph $G = (V,E)$ of which
the edges in $E$ arrive online. There is an edge weight function $w :
2^E \rightarrow \R_{\geq 0}$ and a weight is revealed if an edge
arrives. The goal is to select a forest (i.e., a subset of edges that
does not give rise to a cycle) of maximum weight. The possible forests
of $G$ form the \emph{independent sets of the graphical matroid on
  $G$}. It is well-known that the offline optimal solution of this
problem can be found by the greedy algorithm that orders all the edge weights in decreasing order, and selects elements in this order whenever possible.

We will next explain the predictions that we consider in this
section. For every node $v \in V$, we let $p_v^*$ be a prediction for
the \emph{maximum edge weight $\max_{u \in \mathcal{N}(v)} w_{uv}$} adjacent to $v \in V$. The prediction error is defined  by
$$
\eta = \max_{v \in V} \big|p_v^* - w_{\max}(v)\big|,\quad\text{where}\quad w_{\max}(v)=\max_{u \in \mathcal{N}(v)} w_{uv}.
$$

\begin{remark}
Although the given prediction is formulated independently of any
optimal solution (as we did in the previous sections), it is
nevertheless equivalent to a prediction regarding the maximum weight
$w_{uv}$ adjacent to $v \in V$ in an optimal (greedy) solution. To see
this, note that the first time the offline greedy algorithm encounters
an edge weight $w_{uv}$ adjacent to $v$, it can always be added as currently there is no edge adjacent to $v$. So adding the edge $\{u,v\}$ cannot create a cycle.
\end{remark}

Before we give the main result of this section, we first provide a \emph{deterministic} $(1/4 - o(1))$-competitive algorithm for the graphic matroid secretary problem in Section \ref{sec:det_graphic}, which is of independent interest. Asymptotically, this is an improvement over the algorithm of Soto et al. \cite{STV2018} in the sense that we do not require randomness in our algorithm. We then continue with an algorithm incorporating the predictions in Section \ref{sec:predic_graphic}.

\subsection{Deterministic approximation algorithm}\label{sec:det_graphic}
In this section, we provide a deterministic $(1/4 -
o(1))$-approximation for the graphic matroid secretary problem. For a
given undirected graph $G = (V,E)$, we use the bipartite graph
interpretation that was also used in  \cite{BIKK2018}. That is, we
consider the bipartite graph $B_G = (E \cup V, A)$, where an edge
$\{e,v\} \in A$, for $e \in E$ and $v \in V$, if and only if $v \in
e$. Note that this means that every $e = \{u,v\}$ is adjacent to
precisely $u$ and $v$ in the bipartite graph $B_G$. Moreover, the edge
weights for $\{e,v\}$ and $\{e,u\}$ are both $w_e$ (which is revealed upon arrival of the element $e$).\footnote{We call the edges in $E$ (of the original graph $G$) elements, in order to avoid confusion with the edges of $B_G$.} We emphasize that in this setting, the $e \in E$ are the elements that arrive online.

Algorithm \ref{alg:graphic} is very similar to the algorithm in \cite{KRTV2013} with the only difference that we allow an edge $\{e,u\}$ or $\{e,v\}$ to be added to the currently selected matching $M$ in $B_G$ if and only if \emph{both} nodes $u$ and $v$ are currently not matched in $M$. In this section we often represent a (partial) matching in $B_G$ by a directed graph (of which its undirected counterpart does not contain any cycle). In particular, given some matching $M$ in $B_G$, we consider the directed graph $D_M$ with node set $V$. There is a directed edge $(u,v)$ if and only if $\{e,v\}$ is an edge in $M$, where $e = \{u,v\} \in E$. Note that every node in $D_M$ has an in-degree of at most one as $M$ is a matching.

Using the graph $D_M$ it is not too hard to see that if both $u$ and $v$ are not matched in the current matching $M$, then adding the edge $\{e,u\}$ or $\{e,v\}$ can never create a cycle in the graph formed by the elements $e \in E$ matched up by $M$, called $E[M]$, which is the currently chosen independent set in the graphic matroid. This subgraph of $G$ is precisely the undirected counterpart of the edges in $D_M$ together with $\{u,v\}$. For sake of contradiction, suppose adding the $\{u,v\}$ to $E[M]$ would create an (undirected) cycle $C$. As both $u$ and $v$ have in-degree zero (as they are unmatched in $M$), it follows that some node on the cycle $C$ must have two incoming directed edges in the graph $D_M$. This yields a contradiction.

We note that, although $u$ and $v$ being unmatched is sufficient to guarantee that the edge $\{u,v\}$ does not create a cycle, this is by no means a necessary condition.

\IncMargin{1em}
\begin{algorithm}
\SetKwData{Left}{left}\SetKwData{This}{this}\SetKwData{Up}{up}
\SetKwFunction{Union}{Union}\SetKwFunction{FindCompress}{FindCompress}
\SetKwInOut{Input}{Input}\SetKwInOut{Output}{Output}
\Input{Bipartite graph $G_B = (E \cup V, A)$ for undirected weighted graph $G = (V,E)$ with $|E| = m$.}
\Output{Matching $M$ of $G_B$ corresponding to forest in $G$.}\medskip

\textbf{Phase I:} \hfill  \\
\For{$i = 1,\dots,\lfloor m/c \rfloor$}{
Observe arrival of element $e_i$, but do nothing.}
Let $E' = \{e_1,\dots,e_{\lfloor m/c \rfloor}\}$ and $M = \emptyset$.\\
\textbf{Phase II:} \\
\For{$i = \lfloor m/c \rfloor + 1, \dots, m$}{
Let $E' = E' \cup e_i$.\\
Let $M^{i} = \text{optimal matching on } B_G[E' \cup V]$.\\
Let $a^{i} = \{e_i,u\}$ be the edge assigned to $e_i = \{u,v\}$ in $M^{i}$ (if any).\\
\If{$M \cup a^i$ is a matching and both $u$ and $v$ are unmatched in $M$}{
Set $M = M \cup a^i$.}}
\caption{Deterministic graphic matroid secretary algorithm}
\label{alg:graphic}
\end{algorithm}
\DecMargin{1em}

\begin{theorem}\label{thm:detMatroid}
Algorithm \ref{alg:graphic} is a deterministic $(1/4 - o(1))$-competitive algorithm for the graphic matroid secretary problem.
\end{theorem}

\subsection{Algorithm including predictions}\label{sec:predic_graphic}
In this section we will augment Algorithm \ref{alg:graphic} with the
predictions for the maximum edge weights adjacent to the nodes in $V$. We will use the bipartite graph representation $B_G$ as introduced in Section \ref{sec:graphic}. Algorithm \ref{alg:graphic_predic} consists of three phases, similar to Algorithm \ref{alg:bipartite}. 

Instead of exploiting the predictions in Phase III, we already exploit them in Phase II for technical reasons.\footnote{One could have done the same in Algorithm \ref{alg:bipartite}, but this leads to a worst-case bound that is worse than $\ln(c/d)/c$.} Roughly speaking, in Phase II, we run a greedy-like algorithm that selects for every node $v \in V$ at most one edge that satisfies a threshold based on the prediction for node $v$. In order to guarantee that we do not select too many edges when the predictions are poor (in particular when they are too low), we also include a `'fail-safe' threshold based on the edges seen in Phase I.

\IncMargin{1em}
\begin{algorithm}[t]
\SetKwData{Left}{left}\SetKwData{This}{this}\SetKwData{Up}{up}
\SetKwFunction{Union}{Union}\SetKwFunction{FindCompress}{FindCompress}
\SetKwInOut{Input}{Input}\SetKwInOut{Output}{Output}
\Input{Bipartite graph $G_B = (E \cup V, A)$ for undirected graph $G = (V,E)$ with $|E| = m$. Predictions $p = (p_1^*,\dots,p_n^*)$. Confidence parameter $0 \leq \lambda \leq \min_i p_i^*$ and $c > d \geq 1$.}
\Output{Matching $M$ of $G_B$ corresponding to forest in $G$.}\medskip
\textbf{Phase I:} \hfill  \\
\For{$i = 1,\dots,\lfloor m/c \rfloor$}{
Let $e_i = \{u,v\}$.\\
Set $t_v = \max\{t_v,w(u,v)\}$ and $t_u = \max\{t_u,w(u,v)\}$.}
Let $E' = \{e_1,\dots,e_{\lfloor m/c \rfloor}\}$ and $M = \emptyset$.\\
\textbf{Phase II: }\\
\For{$i = \lfloor m/c \rfloor + 1, \dots, \lfloor m/d \rfloor$}{
Let $e_i = \{u,v\}$, $S = \{x \in \{u,v\} : x \notin E[M] \text{ and } w(u,v) \geq \max \{t_x, p_x^* - \lambda\}\}$ and $y_i = \text{argmax}_{x \in S} p_x^* - \lambda$.\\
\If{$E[M] \cup \{e_i\}$ does not contain a cycle}
{
Set $M = M \cup \{e_i,y_i\}$.
}
}
\textbf{Phase III:} \\
\For{$i = \lfloor m/d \rfloor + 1, \dots, m$}{
Let $E' = E' \cup e_i$.\\
Let $M^{i} = \text{optimal matching on } B_G[E' \cup V]$.\\
Let $a^{i} = \{e_i,u\}$ be the edge assigned to $e_i = \{u,v\}$ in $M^{i}$ (if any).\\
\If{$M \cup a^i$ is a matching and both $u$ and $v$ are unmatched in $M$}{
Set $M = M \cup a^i$.}}
\caption{Graphic matroid secretary algorithm with predictions}
\label{alg:graphic_predic}
\end{algorithm}
\DecMargin{1em}

\begin{theorem}\label{thm:graphic}
	For any $\lambda \ge 0$ and $c>d\ge 1$,
	there is a deterministic algorithm for the graphic matroid secretary
	problem that is asymptotically $g_{c,d,\lambda}(\eta)$-competitive in
	expectation, where
	\[
	g_{c,d,\lambda}(\eta) = \left\{
	\begin{array}{ll}
	\max\left\{\frac{d-1}{c^{2}},\frac{1}{2}\left(\frac{1}{d} - \frac{1}{c} \right)\left(1 - \frac{2(\lambda + \eta)|V|}{\text{OPT}}\right) \right\} & \text{ if } 0 \leq \eta < \lambda,\\
	\frac{d-1}{c^{2}} & \text{ if } \eta \geq \lambda.
	\end{array}\right.
	\]
\end{theorem}
If $\lambda$ is small, we roughly obtain a bound of $(1/d - 1/c)/2$ in case $\eta$ is small, and a bound of $(d-1)/c^2$ if $\eta$ is large. Note that the roles of $c$ and $d$ have interchanged w.r.t. Algorithm \ref{alg:bipartite} as we now exploit the predictions in Phase II instead of Phase III. Roughly speaking, if $d \rightarrow 1$ and $c \rightarrow \infty$ we approach a bound of $1/2$ if the predictions are good, whereas the bound of $(d-1)/c^2$ becomes arbitrarily bad.

\begin{proof}[Proof of Theorem \ref{thm:graphic}]
As in the proof of Theorem \ref{thm:bipartite}, we provide two lower bounds on the expected value on the matching $M$ outputted by Algorithm \ref{alg:graphic_predic}. We first provide a bound of 
$$
\frac{1}{2}\left(\frac{1}{d} - \frac{1}{c} \right)\left(1 - \frac{2(\lambda + \eta)|V|}{\text{OPT}}\right)\text{OPT}
$$
in case the prediction error is small, i.e., when $\eta < \lambda$. For simplicity we focus on the case where for each $v \in V$, the weights of the edges adjacent to $v$ in $G$ are distinct.\footnote{For the general case, one can use a global ordering on all edges in $E$ and break ties where needed.}

For every $v \in V$, let $e_{max}(v)$ be the (unique) edge adjacent to $v$ with maximum weight among all edges adjacent to $v$. 
Consider the fixed order $(e_1,\dots,e_m)$ in which the elements in $E$ arrive online, and define $Q = \{v \in V : e_{max}(v) \text{ arrives in Phase II}\}$. We will show that the total weight of all edges selected in Phase II is at least
$
\frac{1}{2}\sum_{v \in Q} (p_v^* - (\lambda + \eta)).
$
Let $T \subseteq Q$ be the set of nodes for which the edge $e_{\max}(v)$ arrives in Phase II, but for which $v$ does not get matched up in Phase II.

In particular, let $v \in T$ and consider the step $\ell$ in Phase II in which $e_{\max}(v) = \{u,v\}$ arrived.
By definition of $\eta$, and because $\eta < \lambda$, we have
\begin{equation}\label{eq:wuv_emaxV}
w(u,v) = w_{\max}(v) \geq \max\{t_v, p_v^* - \eta\} \geq \max\{t_v, p_v^* - \lambda\},
\end{equation}
and so the pair $\{e_i,v\}$ is eligible (in the sense that $v \in S$). Since $v$ did not get matched, one of the following two holds:
\begin{enumerate}[i)]
\item The edge $e_{\max}(v)$ got matched up with $u$.
\item Adding the edge $\{e_{\max}(v),v\}$ to $M$ would have yielded a cycle in $E[M] \cup e_{\max}(v)$.
\end{enumerate}
Note that it can never be the case that we do \emph{not} match up $e_{\max}(v)$ to $u$ for the reason that it would create a cycle. This is impossible as both $u$ and $v$ are unmatched.

Now, in the first case, since $u$ is matched it must hold that 
$w(u,v)\geq\max\{t_u,p_u^*-\lambda\}$, and $p_{u}^* - \lambda \geq p_v^* - \lambda$
as $v$ was eligible to be matched up in the online matching $M$ (but it did not happen).
Further, combining \eqref{eq:wuv_emaxV} and the definition of $y_i$ in Phase II, yields
\begin{equation}
2w(u,v) \geq (p_{u}^* - \lambda) + (p_{v}^* - \eta) \geq (p_{u}^* - \eta - \lambda) + (p_{v}^* - \eta - \lambda).
\label{eq:in1}
\end{equation}
We call $u$ the \emph{i)-proxy} of $v$ in this case.

In the second case, if adding $e_{\max}(v)$ would have created an (undirected) cycle in the set of elements (i.e., the forest) selected so far, this yields a unique directed cycle in the graph $D_M$ defined in the previous section. If not, then there would be a node with two incoming arcs in $D_M$, as every arc on the cycle is oriented in some direction. This would imply that $M$ is not a matching.

Let $e' = \{u,z\} \in E$ be the element corresponding to the incoming arc at $u$ in $D_M$. Note that by assumption $u$ is already matched up, as $e_{\max(v)}$ creates a directed cycle in $D_{M \cup e_{\max}(v)}$. That is, we have $\{e',u\} \in M$. 
Then, by definition of $\eta$, we have
\begin{equation}\label{eq:pu_pv_eta}
\eta+p_u^*  \geq e_{\max}(u) \geq w(u,v) \geq p_v^* -\eta,
\end{equation}
where the last inequality holds by \eqref{eq:wuv_emaxV}.
Combining~\eqref{eq:pu_pv_eta} with the fact that $w(u,z)\geq p_u^*-\lambda$ (because $\{u,z\}$ got matched up to $u$), and the fact that 
$e_{\max}(v)\geq p_v^* - \eta$, by \eqref{eq:wuv_emaxV}, it follows that
\begin{equation}\label{eq:in2}
2w(u,z) \geq [p_{u}^* - (\lambda + \eta)] + [p_{v}^* - (\lambda + \eta)].
\end{equation}
In this case, we call $u$ the \emph{ii)-proxy} of $v$.\\

\begin{claim}
For any distinct $v, v' \in T$, their corresponding proxies $u$ and $u'$ are also distinct.
\end{claim}
\begin{proof}
Suppose that $u = u'$. The proof proceeds by case distinction based on the proxy types.
\begin{enumerate}
	\item $u = u'$ is defined as i)-proxy for both $v$ and $v'$: This cannot happen as $u = u'$ would then have been matched up twice by Algorithm \ref{alg:graphic_predic}.	
	\item $u = u'$ is defined as ii)-proxy for both $v$ and $v'$: In this case there is a directed cycle with the arc $(u,v) = (u',v)$ and 
	another directed cycle with the arc $(u',v') = (u,v')$. Hence, there is a vertex with two incoming arcs in $D_M$.
	This also means that Algorithm \ref{alg:graphic_predic} has matched up a vertex twice, which is not possible.	
	\item $u = u'$ is defined as i)-proxy for $v$ and as ii)-proxy for $v'$: Then $e_{\max}(v)$, which gets matched up with $u = u'$, must have arrived before $e_{\max}(v')$. If not, then both $v'$ and $u = u'$ would have been unmatched when $e_{\max}(v')$ arrived and we could have matched it up with at least $v'$ (as this cannot create a cycle since $u = u'$ is also not matched at that time). This means that when $e_{\max}(v')$ arrived, the reason that we did not match it up to $v'$ is because this would create a directed cycle in $D_{M \cup e_{\max}(v')}$. But, as $u$ has an incoming arc from $v$ in $D_M$, this means that the directed cycle goes through $v$, which implies that $v$ \textit{did} get matched up in Phase II, which we assumed was not the case.

\end{enumerate}
This concludes the proof of the claim. \label{claim:proxy}
\end{proof}
Using Claim \ref{claim:proxy} in combination with \eqref{eq:in1} and \eqref{eq:in2}, we then find that 
$$
w[M_{II}] \geq \frac{1}{2}\sum_{v \in Q} [p_v^* - (\lambda + \eta)],
$$
where $M_{II}$ contains all the edges obtained in Phase II. Roughly speaking, for every edge $e_{\max}(v)$ that we cannot select in Phase II, there is some other edge selected in Phase II that `covers' its weight in the summation above (and for every such $v$ we can find a unique edge that has this property).

Now, in general, we have a uniformly random arrival order, and therefore, for every $v \in V$, the probability that edge $e_{\max}(v)$ arrives in Phase II equals
$
\frac{1}{d} - \frac{1}{c}.
$
Therefore, with expectation taken over the arrival order, we have
$$
\mathbb{E}[M_{II}] \geq \frac{1}{2}\left(\frac{1}{d} - \frac{1}{c} \right) \sum_{v \in V} (p_v^* - (\lambda+\eta)) \geq \frac{1}{2}\left(\frac{1}{d} - \frac{1}{c} \right)\left(1 - \frac{2(\lambda + \eta)|V|}{\text{OPT}}\right)\text{OPT}.
$$\medskip

\noindent We continue with the worst-case bound that holds even if the prediction error is large. We first analyze the probability that two given nodes $u$ and $v$ do not get matched up in Phase II. Here, we will use the thresholds $t_v$ defined in Algorithm \ref{alg:graphic_predic}. 

Conditioned on the set of elements $A$ that arrived in Phase I/II, 
the probability that the maximum edge weight adjacent to $v$, over all edges adjacent to $v$ in $A$, 
appears in Phase I is at equal to $(d/c)$.
This implies that $v$ will not get matched up in Phase II, 
by definition of Algorithm \ref{alg:graphic_predic}.
The worst-case bound of $(d-1)/c^2$ is proven in Appendix \ref{app:graphic_prediction}.
\end{proof}

%% file: conclusion.tex
\section{Conclusion}
Our results can be seen as the first evidence that online selection problems are a promising area
for the incorporation of machine learned advice following the frameworks of  \cite{LykourisV18,PurohitSK18}. Many interesting problems and directions remain open. 
For example, does there exist a natural prediction model for the general matroid secretary problem? It is 
still open whether this problem admits a constant-competitive algorithm. Is it possible to show that there exists 
an algorithm under a natural prediction model that is
constant-competitive for accurate predictions, and that is still
$O(1/\log(\log(r)))$-competitive in the worst case, matching the
results in \cite{FSZ2014,L2014}? 
Furthermore, although our results are optimal within
	this specific three phased approach, it remains an open question whether
	they are optimal in general for the respective problems.

\paragraph*{Acknowledgments:}
We would like to thank Kurt Mehlhorn for drawing our attention to
  the area of algorithms augmented with machine learned predictions.

%% file: appendix_secretary.tex
\section{Secretary Problem}\label{app:secProb}

This section is organized as follows:

In Subsection~\ref{app:randAlgo}, we show that the discontinuity at $\eta = \lambda$, in Theorem~\ref{thm:classical}, can be smoothed out by selecting $\lambda$ according 
to some distribution with mean representing our confidence in the prediction $p^*$.
Further, we study the competitive ratio as a function of the prediction error $\eta$.

In Subsection~\ref{app:CompNRndAlgo}, we demonstrate that our deterministic Algorithm~\ref{alg:classical}
can achieve a better competitive ratio than its corresponding naive randomization.

\subsection{Randomized Algorithms}\label{app:randAlgo}
We note that Algorithm~\ref{alg:classical} although deterministic, can
be relatively easily transformed to an algorithm that picks the
confidence parameter $\lambda\in[0,p^*]$ according to some probability
distribution. Algorithm~\ref{alg:classical} is then the special case
where the whole probability mass of the distribution is at one point
in $[0,p^*]$.

This naturally gives rise to the question of whether there exists a
distribution that outperforms the deterministic algorithm. It can be
relatively easily seen, that the deterministic algorithm with $\lambda
= \eta$ is the best possible competitive ratio that can be obtained
with such an approach. Therefore, a randomized algorithm can at best
match the deterministic one with $\lambda = \eta$, and this happens
only in the case in the center of mass of the used distribution is at
$\eta$. Since $\eta$ is unknown to the algorithm, it is not possible
to select a distribution that outperforms any deterministic algorithm
for all possible $\eta$'s.

Despite that, it may still be advantageous to pick $\lambda$ at random
according to some distribution, in order to avoid the ``jump'' that
occurs at $\eta = \lambda$ in the competitive ratio of the
deterministic algorithm. In particular for an appropriate distribution
with density function $h_\lambda(x)$ the expected competitive ratio is
given by:
\begin{align*}
	\mathbb{E}_\lambda\left[g_{c,\lambda}(\eta)\right] =
	Pr[\lambda<\eta]\cdot\frac{1}{ce} + f(c)\int_\eta^{p^*}h_\lambda(x)\left(1-\frac{x+\eta}{OPT}\right)dx,
\end{align*}
which can be seen as a convex combination of two competitive
ratios. 

Some example distributions and how they compare to an algorithm that
selects $\lambda$ deterministically can be seen in Figure~\ref{fig:test}.

\begin{figure}[h]
	\centering
	\begin{subfigure}{0.5\textwidth}
		\centering
		\includegraphics[width=0.8\linewidth, bb=0 0 1100 850]{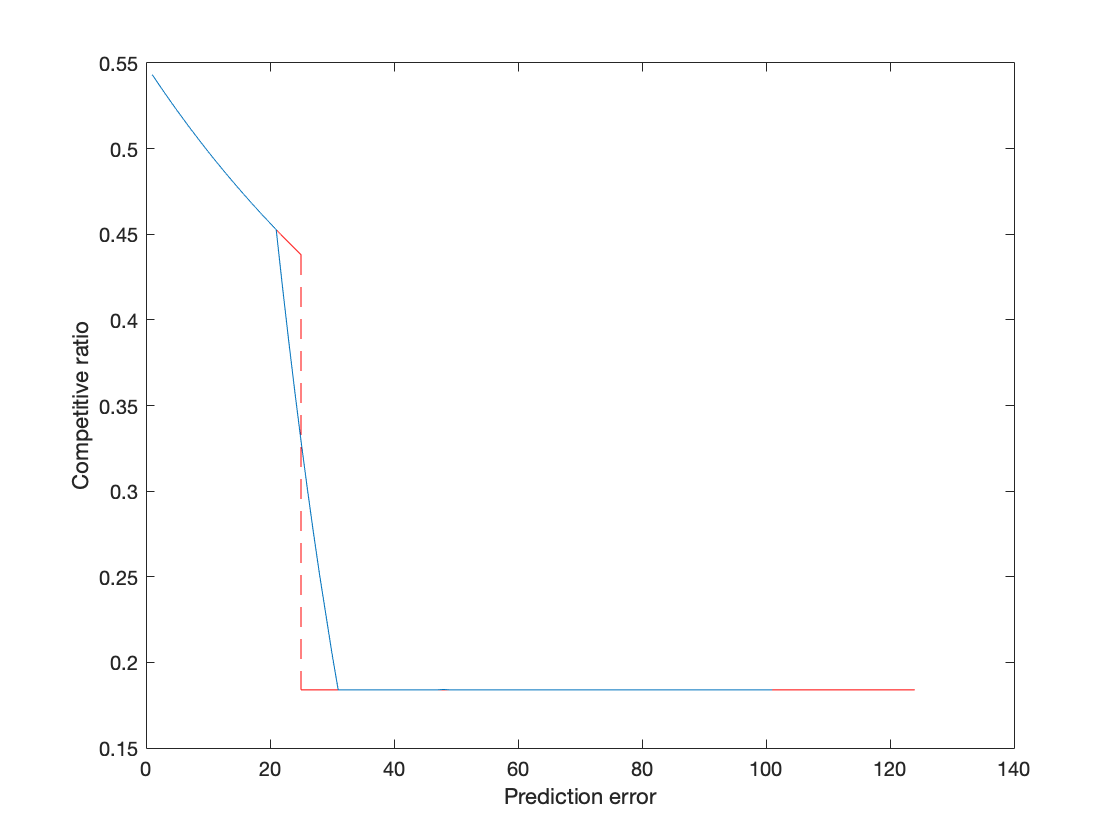}
		\caption{}
		\label{fig:sub1}
	\end{subfigure}%
	\begin{subfigure}{0.5\textwidth}
		\centering
		\includegraphics[width=0.8\linewidth, bb=0 0 1100 850]{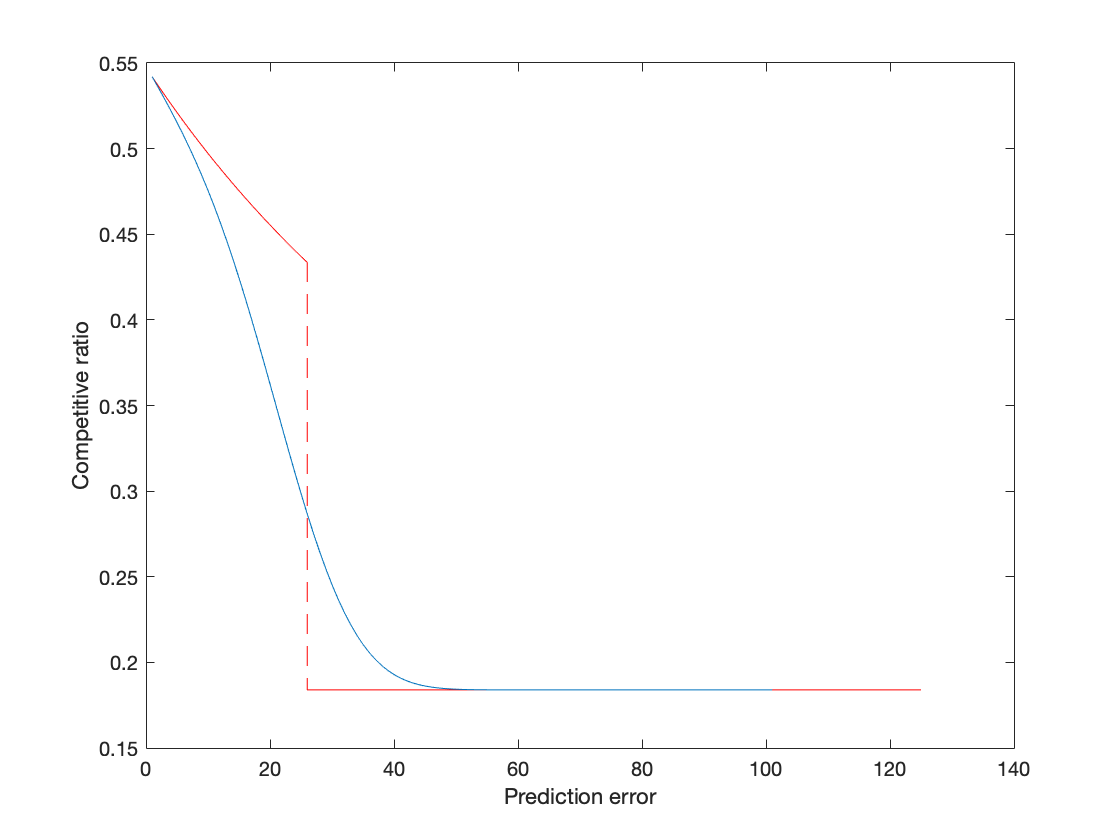}
		\caption{}
		\label{fig:sub2}
	\end{subfigure}\\
	\begin{subfigure}{0.5\textwidth}
		\centering
		\includegraphics[width=0.8\linewidth, bb=0 0 1100 850]{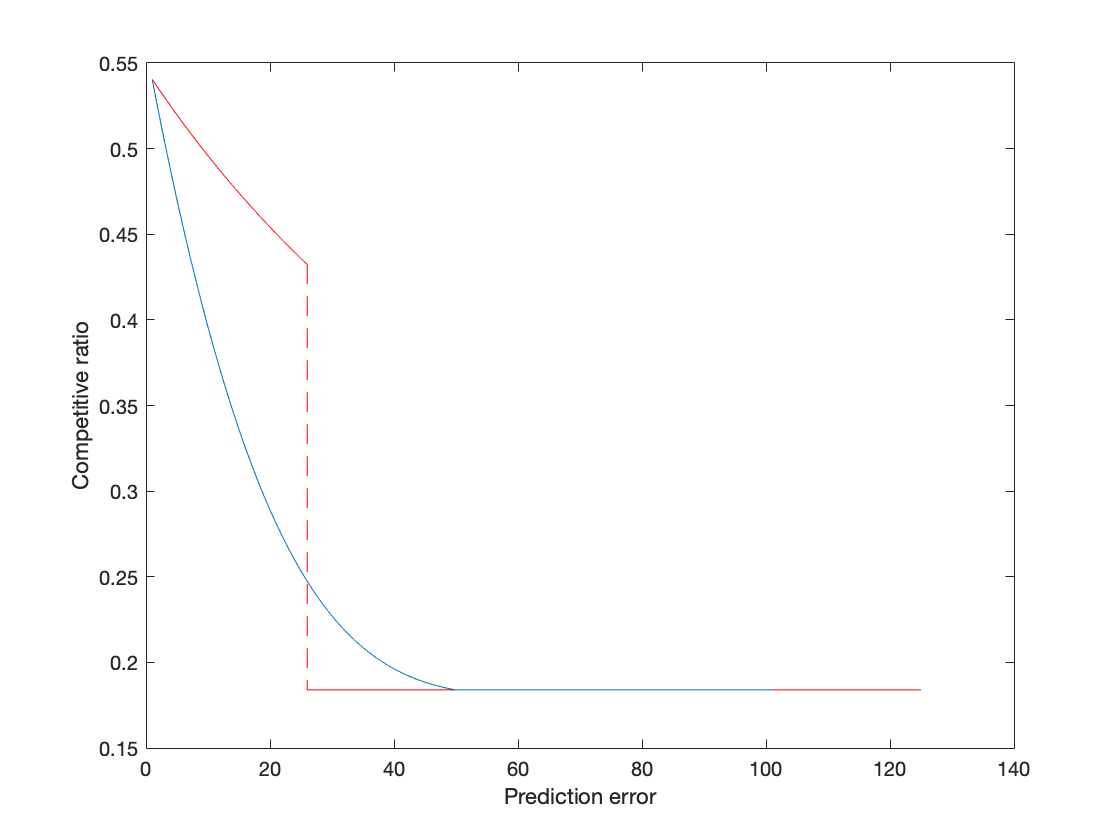}
		\caption{}
		\label{fig:sub3}
	\end{subfigure}%
	\caption{A comparison the deterministic (in red) and
		randomized (in blue) choice
		of $\lambda$. The $y$-axes are the competitive ratio, the
		$x$-axes are the prediction error $\eta$, and all figures consider
		$p^*=100$. In Subfigure~\ref{fig:sub1} we take $\lambda = 25$ for
		the deterministic algorithm and choose $\lambda$ according to the
		uniform distribution in $(20,30)$ for the randomized one. In
		Subfigure~\ref{fig:sub2} we have $\lambda\approx 25$ for the deterministic
		one and the normal distribution with mean $25$ and variance
		$10$. Finally, in Subfigure~\ref{fig:sub3} we have $\lambda \approx 25$, and
		the normal distribution with mean $0$ and variance $32$. All plots
		have $c=2$.}
	\label{fig:test}
\end{figure}

\subsection{Comparison between Algorithm~\ref{alg:classical} and its naive randomization}\label{app:CompNRndAlgo}

One natural question that arises from the bound in Theorem \ref{thm:classical} is whether one can significantly
improve the result using randomization. 

Here, we provide a brief comparison with the following naive randomization of Algorithm~\ref{alg:classical},
which randomly chooses between running the classical secretary problem without predictions
and (roughly speaking) the greedy prediction-based procedure in Phase II in Algorithm~\ref{alg:classical}.
That is, given $\gamma \in [0,1]$, with probability $\gamma$ it runs the classical secretary problem, 
and with probability $1 - \gamma$, it runs the prediction-based algorithm that simply selects 
the first element with value greater or equal than $p^* - \lambda$ (if any).
Note that its expected competitive ratio at least
\begin{equation}\label{eq:obvious}
\gamma \frac{1}{e} + (1 - \gamma)\left(\max\left\{1 - \frac{\lambda + \eta}{\text{OPT}},0\right\}\right).
\end{equation}

In order to compare Algorithm~\ref{alg:classical} with its naive randomization, we set $\gamma = 1/c$.
This implies that both algorithms are at least $1/(ce)$-competitive in the worst-case when 
the predictions are poor. Having the same worst-case guarantee, we now focus on their performance 
in case when the prediction error $\eta$ and the confidence parameter $\lambda$ are small. 
In particular, let us consider the case where 
$\lambda + \eta = \delta\cdot\text{OPT}$ for some small $\delta > 0$. 
Then, the expected competitive ratio in \eqref{eq:obvious} reduces to
\begin{equation}\label{eq:naiveRandAlglCR}
\frac{1}{ce} + \left(1 - \frac{1}{c}\right)(1 - \delta).
\end{equation}

We now compare the expected competitive ratio of Algorithm~\ref{alg:classical}
and its naive randomization, which read $f(c)(1- \delta)$ and \eqref{eq:naiveRandAlglCR} 
respectively.
In Figure~\ref{fig:OblRndvsOurDet}, we conduct a numerical experiment with fixed
$\delta = 0.1$ and $\lambda + \eta = 0.1\text{OPT}$.
Our experimental data indicates that for $c\geq1.185$,
Algorithm~\ref{alg:classical} is at least $1/e$-competitive and it significantly 
outperforms the classical secretary algorithm as $c$ increases. 
Furthermore, for $c \geq 1.605$ Algorithm~\ref{alg:classical} performs better than its naive randomization.
On the other hand, we note that our experiments indicate that
as $\delta$ increases the competitive advantage of Algorithm~\ref{alg:classical} 
over its naive randomization decreases.

\begin{figure}
	\begin{center}
		\includegraphics[width=0.6\textwidth, bb=0 0 500 380]{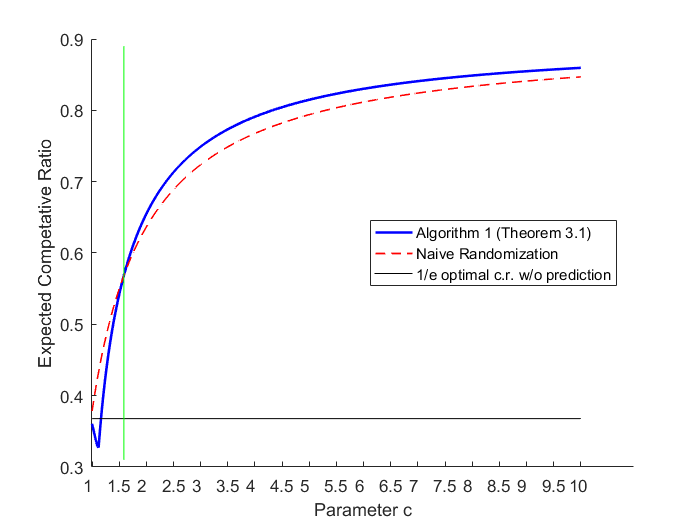}
	\end{center}
	\caption{The black horizontal line indicates the  tight bound of $1/e$ for the classical secretary algorithm. The bold blue line is the performance guarantee for Algorithm~\ref{alg:classical}; and the dashed red line is the performance guarantee for the obvious randomized algorithm.}
	\label{fig:OblRndvsOurDet}
\end{figure}

%% file: appendix_kesselheim.tex
\section{Perfect Matching Instances}

Given an undirected weighted bipartite graph $G^{\prime}=(L^{\prime}\cup R^{\prime},E^{\prime},w^{\prime})$
we construct an augmented bipartite graph $G=(L\cup R,E,w)$ as follows:

i) the left node set $L=L^{\prime}$; ii) the right node set $R=R^{\prime}\cup L^{\prime}$;
iii) the edge set $E=E^{\prime}\cup F^{\prime}$ where the set $F^{\prime}$
consists of edges $\{u_{i},v_{i}\}$ such that $u_{i}$ and $v_{i}$ are
the $i$-th node in $L$ and $L^{\prime}$ respectively, for all $i\in\{1,\dots,|L|\}$;
iv) $w(e)=w^{\prime}(e)$ for all edges $e\in E^{\prime}$ and $w(e)=0$
for all edges $e\in F^{\prime}$.

We call the resulting bipartite graph $G$ \emph{perfect}.
\begin{fact}\label{fact:pBG} 
	Suppose $G=(L\cup R,E,w)$ is a perfect bipartite
	graph. Let $\ell\in\{1,\dots,|L|\}$ be an arbitrary index and $\mathcal{L}(\ell)=\left\{ S\subseteq L\ :\ |S|=\ell\right\} $
	be the set of all subsets of nodes in $L$ of size $\ell$. Then,
	for every subset $S\in\mathcal{L}(\ell)$ the induced subgraph $G[S\cup N(S)]$
	has a perfect matching $M_{S}$ of size $\ell$, i.e., $|M_{S}|=|S|=\ell$.
\end{fact}

\newpage

\section{General analysis of the algorithm of Kesselheim et al. \cite{KRTV2013}} \label{app:kesselheim}
In this section, we analyze a modified version of the algorithm of Kesselheim et al. \cite{KRTV2013}, 
see Algorithm~\ref{alg:kesselheim}. Our analysis extends the proof techniques presented in
\cite[Lemma 1]{KRTV2013}.

\begin{theorem}
	\label{thm:main} Given a perfect bipartite graph, Algorithm~\ref{alg:kesselheim} is $(\frac{1}{c}-\frac{1}{n})\ln\frac{c}{d}$ competitive in expectation. 
	In addition, the expected weighted contribution of the nodes $\{\lfloor n/d\rfloor+1,\dots,n\}$ 
	to the online matching $M$ is $OPT\cdot(\frac{1}{c}-\frac{1}{n})\ln d$.
\end{theorem}

For convenience of notation, we will number the vertices in $L$ from
$1$ to $n$ in the random order they are presented to the algorithm.
Hence, we will use the variable $\ell$ as an integer, the
name of an iteration and the name of the current node (the last
so far).\\

\IncMargin{1em}
\begin{algorithm}[H]
	\SetKwData{Left}{left}\SetKwData{This}{this}\SetKwData{Up}{up}
	\SetKwFunction{Union}{Union}\SetKwFunction{FindCompress}{FindCompress}
	\SetKwInOut{Input}{Input}\SetKwInOut{Output}{Output}
	\Input{Vertex set $R$ and cardinality $|L| = n$.}
	\Output{Matching $M$.}\medskip
	
	\textbf{Phase I:} \\
	\For{$\ell = 1,\dots,\lfloor n/c \rfloor$}{
		Observe arrival of node $\ell$, but do nothing.}
	Let $L' = \{1,\dots,\lfloor n/c \rfloor\}$ and $M = \emptyset$.\\
	\textbf{Phase II:} \\
	\For{$\ell = \lfloor n/c \rfloor + 1, \dots, \lfloor n/d \rfloor$}{
		Let $L' = L' \cup \ell$.\\
		Let $M^{(\ell)} = \text{optimal matching on } G[L' \cup R]$.\\
		Let $e^{(\ell)} = (\ell,r)$ be the edge assigned to $\ell$ in $M^{(\ell)}$.\\
		\If{$M \cup e^{(\ell)}$ is a matching}{
			Set $M = M \cup \{e^{(\ell)}\}$.}}
	\caption{Online bipartite matching algorithm (under uniformly random vertex arrivals)}
	\label{alg:kesselheim}
\end{algorithm}
\DecMargin{1em}

\subsubsection*{Organization}

In Subsection \ref{sub:Notation}, we present the notation. In Subsection
\ref{sub:SL}, we give the main structural result and prove Theorem
\ref{thm:main}. In addition, in Subsection \ref{sub:PhaseIII}, we
give a lower bound on the probability that an arbitrary node $r\in R$
remains unmatched after the completion of Phase II.

\subsection{Notation \label{sub:Notation}}

Consider the following random process:

Sample uniformly at random a permutation of the nodes $L$. Let $L_{\ell}$
be a list containing the first $\ell$ nodes in $L$, in the order
as they appear, and let $M^{(\ell)}$ be the corresponding optimum
matching of the induced subgraph $G^{(\ell)}$ on the node set $L_{\ell}\cup N(L_{\ell})$.

Let $E^{(\ell)}$ be the event \{$e^{(\ell)}\cup M$ is a matching\},
where (r.v.) $M$ is the current online matching. Note that the existence
of edge (r.v.) $e^{(\ell)}$ is guaranteed by the Fact \ref{fact:pBG} and $G$
is a perfect bipartite graph. We define a random variable
\[
A_{\ell}=\begin{cases}
w(e^{(\ell)}) & ,\text{ if event }E^{(\ell)}\text{ occur};\\
0 & ,\text{ otherwise.}
\end{cases}
\]

\subsection{Structural Lemma \label{sub:SL}}
\begin{lemma}
	\label{lem:Exp} Suppose $G=(L\cup R,E,w)$ is a perfect bipartite
	graph. Then, for every $c>1$ it holds for every $\ell\in\{\lfloor n/c \rfloor + 1,\dots,n\}$
	that
	\[
	\mathbb{E}[A_{\ell}]\geq\frac{\lfloor n/c\rfloor}{n}\cdot\frac{OPT}{\ell-1}.
	\]
	
\end{lemma}

Before we prove Lemma \ref{lem:Exp}, we show that it implies Theorem
\ref{thm:main}.

\subsubsection{Proof of Theorem \ref{thm:main} \label{sub:PT1}}

Using Lemma \ref{lem:Exp}, we have
\[
\mathbb{E}\left[\sum_{\ell=1}^{n/d}A_{\ell}\right]=\sum_{\ell=\lfloor n/c\rfloor+1}^{n/d}\mathbb{E}\left[A_{\ell}\right]\geq\sum_{\ell=\lfloor n/c\rfloor+1}^{n/d}\frac{\lfloor n/c\rfloor}{n}\cdot\frac{OPT}{\ell-1}\geq OPT\cdot\left(\frac{1}{c}-\frac{1}{n}\right)\cdot\ln\frac{c}{d},
\]where the inequalities follow by combining $\frac{\lfloor n/c\rfloor}{n}\geq\frac{1}{c}-\frac{1}{n}$
and
\[
\sum_{\ell=\lfloor n/c\rfloor+1}^{n/d}\frac{1}{\ell-1}=\sum_{\ell=\lfloor n/c\rfloor}^{n/d-1}\frac{1}{\ell}\geq\ln\frac{n/d}{\lfloor n/c\rfloor}\geq\ln\frac{c}{d}.
\]

\subsubsection{Proof of Lemma \ref{lem:Exp}}

We prove Lemma \ref{lem:Exp} in two steps. Observe that $\mathbb{E}[A_{\ell}\ |\invneg E^{(\ell)}]=0$ implies
\[
\mathbb{E}\left[A_{\ell}\right]=\mathbb{E}\left[w(e^{(\ell)})\ |\ E^{(\ell)}\right]\cdot Pr\left[E^{(\ell)}\right].
\]

We proceed by showing, in Lemma \ref{lem:BipGraph}, that $\mathbb{E}\left[w(e^{(\ell)})\ |\ E^{(\ell)}\right]\geq\frac{OPT}{n}$,
and then in Lemma \ref{lem:PEEex} that $Pr\left[E^{(\ell)}\ |\ E^{(\ell)}\right]\geq\frac{\lfloor n/c\rfloor}{\ell-1}$.\\

Let $S$ be a subset of $L$ of size $\ell$, and let $M_{S}$ be
the optimum weighted matching w.r.t. the induced subgraph $G[S\cup N(S)]$. 
For a fixed subset $S\subseteq L$
with size $\ell$, let $R_{\ell}(S)$ be the event that \{the node set of $L_{\ell}$ equals $S$\}, 
i.e. $\text{Set}(L_{\ell})=S$.
Let $\mathcal{L}(\ell)$ be the set of all subsets of $L$ of size
$\ell$, i.e., $\mathcal{L}(\ell)=\left\{ S\subseteq L\ :\ |S|=\ell\right\} $.
\begin{lemma}\label{lem:BipGraph} 
	For every perfect bipartite graph $G=(L\cup R,E,w)$
	it holds that
	\[
	\mathbb{E}\left[w(e^{(\ell)})\ |\ E^{(\ell)}\right]\geq\frac{OPT}{n}.
	\]
\end{lemma}
\begin{proof}
	Using conditional expectation,
	\begin{equation}
	\mathbb{E}\left[w(e^{(\ell)})\ |\ E^{(\ell)}\right]=\sum_{S\in\mathcal{L}(\ell)}\mathbb{E}\left[w(e^{(\ell)})\ |\ R_{\ell}(S)\wedge E^{(\ell)}\right]\cdot Pr\left[R_{\ell}(S)\right].\label{eq:WweE}
	\end{equation}

	Since the order of $L$ is sampled u.a.r. we have $Pr\left[R_{\ell}(S)\right]=1/{n \choose \ell}$,
	and thus it suffices to focus on the conditional expectation 
	\begin{eqnarray}
	\mathbb{E}\left[w(e^{(\ell)})\ |\ R_{\ell}(S)\wedge E^{(\ell)}\right] & = & \sum_{e^{(i)}\in M_{S}}w(e^{(i)})\cdot Pr_{e^{(\ell)}\sim M_{S}}\left[e^{(\ell)}=e^{(i)}\right]\nonumber \\
	& = & \frac{1}{\ell}\sum_{e^{(i)}\in M_{S}}w(e^{(i)}).\label{eq:EwSE}
	\end{eqnarray}
	where the last equality uses $G$ is a perfect bipartite graph and Fact
	\ref{fact:pBG}. Then, by combining (\ref{eq:WweE},\ref{eq:EwSE})
	we have 
	\begin{equation}
	\mathbb{E}\left[w(e^{(\ell)})\ |\ E^{(\ell)}\right]=\frac{1}{{n \choose \ell}}\cdot\frac{1}{\ell}\sum_{S\in\mathcal{L}(\ell)}\sum_{e^{(i)}\in M_{S}}w(e^{(i)}).\label{eq:EAEexits}
	\end{equation}

	Observe that for any subset $S\subseteq L$, it holds for $M^{\star}|_{S}=\{e^{(i)}=(i,r_{i})\in M^{\star}\ :\ i\in S\}$
	the restriction of the optimum matching $M^{\star}$ (w.r.t. the whole
	graph $G$) on $S$ that 
	\begin{equation}
	\sum_{e^{(i)}\in M_{S}}w(e^{(i)})\geq\sum_{e^{(i)}\in M^{\star}|_{S}}w(e^{(i)}).\label{eq:obs}
	\end{equation}
	Further, since every vertex $i\in L(M^{\star})$ appears in ${n-1 \choose \ell-1}$
	many subsets of size $\ell$ and ${n-1 \choose \ell-1}/{n \choose \ell}=\ell/n$,
	it follows by (\ref{eq:EAEexits},\ref{eq:obs}) that
	\begin{eqnarray*}
		\mathbb{E}\left[A_{\ell}\ |\ E^{(\ell)}\right] & \geq & \frac{1}{{n \choose \ell}}\cdot\frac{1}{\ell}\sum_{S\in\mathcal{L}(\ell)}\sum_{e^{(i)}\in M^{\star}|_{S}}w(e^{(i)})\\
		& = & \frac{{n-1 \choose \ell-1}}{{n \choose \ell}}\cdot\frac{1}{\ell}\sum_{e^{(i)}\in M^{\star}}w(e^{(i)})=\frac{OPT}{n}.
	\end{eqnarray*}
	
\end{proof}

\begin{lemma}
	\label{lem:PEEex} For every perfect bipartite graph $G=(L\cup R,E,w)$
	it holds that
	\begin{equation}
	Pr\left[E^{(\ell)}\right]\geq\frac{\lfloor n/c\rfloor}{\ell-1}.\label{eq:CBB}
	\end{equation}
\end{lemma}
\begin{proof}
	For a fixed subset $S\subseteq L$ with size $\ell$, let $F_{S,\ell}^{(i)}$
	be the event that \{event $R_{\ell}(S)$ occurs\} \emph{and} \{the edge
	$e^{(\ell)}=(i,r_{i})$\}. Then, by conditioning on the choice of
	subset $S\in\mathcal{L}(\ell)$, we have 
	\begin{eqnarray}
	Pr\left[E^{(\ell)}\right] & = & \frac{1}{{n \choose \ell}}\sum_{S\in\mathcal{L}(\ell)}Pr\left[e^{(\ell)}\cup M\text{ is a matching}\ |\ R_{\ell}(S)\right]\nonumber \\
	& = & \frac{1}{{n \choose \ell}}\frac{1}{\ell}\sum_{S\in\mathcal{L}(\ell)}\sum_{(i,r_{i})\in M_{S}}Pr\left[(i,r_{i})\cup M\text{ is a matching}\ |\ F_{S,\ell}^{(i)}\right].\label{eq:PEMEl}
	\end{eqnarray}
	\emph{Note that in (\ref{eq:PEMEl}), the subset $S\subseteq L$ with
		size $\ell$ and the edge $(i,r_{i})\in M_{S}$ are fixed!}
	
	Given the first $k$ nodes of $L$, in the order as they arrive, and
	the corresponding perfect matching (r.v.) $M^{(k)}$, let (r.v.) $e^{(k)}=(k,r_{k})$
	be the edge matched in (r.v.) $M^{(k)}$ from the last node (r.v.)
	$k$. Note that since $G$ is perfect, it follows by Fact \ref{fact:pBG}
	that edge $e^{(k)}$ exists and $|M^{(k)}|=k$. We denote by (r.v.)
	$M^{(k)}[k]$ the corresponding right node $r_{k}$. 
	
	Let $Q_{k}$ be the event that 
	\[
	\{\text{node }r_{i}\not\in M^{(k)}\}\vee\left\{ \{\text{node }r_{i}\in M^{(k)}\}\wedge\{M^{(k)}[k]\neq r_{i}\}\right\} .
	\]
	Then, we have
	\begin{equation}
	Pr\left[(i,r_{i})\cup M\text{ is a matching}\ |\ F_{S,\ell}^{(i)}\right]=Pr\left[\bigwedge_{k=\lfloor n/c\rfloor+1}^{\ell-1}Q_{k}\ |\ F_{S,\ell}^{(i)}\right].\label{eq:fancy}
	\end{equation}
	Observe that the probability of event $\invneg Q_{k}$ is equal to
	\begin{eqnarray}
	&  & Pr\left[\{\text{node }r_{i}\in M^{(k)}\}\wedge\left\{ \{\text{node }r_{i}\notin M^{(k)}\}\vee\{M^{(k)}[k]=r_{i}\}\right\} \right]\nonumber \\
	& = & Pr\left[\{\text{node }r_{i}\in M^{(k)}\}\wedge\{M^{(k)}[k]=r_{i}\}\right].\label{eq:notProb}
	\end{eqnarray}

	Let $W_{i,t}$ be the event that $F_{S,\ell}^{(i)}\wedge\left(\bigwedge_{j=k}^{t-1}Q_{j}\right)$
	for $t\in\{k+1,\dots,\ell-1\}$, and $W_{i,k}$ be the event $F_{S,\ell}^{(i)}$.
	Using conditional probability, we have
	\begin{equation}
	Pr\left[\bigwedge_{k=\lfloor n/c\rfloor+1}^{\ell-1}Q_{k}\ |\ F_{S,\ell}^{(i)}\right]=Pr\left[Q_{\ell-1}\ |\ W_{i,\ell-1}\right]\cdots Pr\left[Q_{k+1}\ |\ W_{i,k+1}\right]\cdot Pr\left[Q_{k}\ |\ W_{i,k}\right].\label{eq:condProb}
	\end{equation}
	We now analyze the terms in (\ref{eq:condProb}) separately. Let $\mathcal{T}(i,k)$
	be the set of all matchings $M^{(k)}$ satisfying the event $F_{S,\ell}^{(i)}\wedge\{\text{node }r_{i}\in M^{(k)}\}$.
	Using (\ref{eq:notProb}), we have
	\begin{eqnarray}
	Pr\left[\invneg Q_{k}\ |\ F_{S,\ell}^{(i)}\right] & = & Pr\left[\{\text{node }r_{i}\in M^{(k)}\}\wedge\{M^{(k)}[k]=r_{i}\}\ |\ F_{S,\ell}^{(i)}\right]\nonumber \\
	& \leq &  Pr\left[M^{(k)}[k]=r_{i}\ |\ F_{S,\ell}^{(i)}\wedge\{\text{node }r_{i}\in M^{(k)}\}\right]\nonumber \\
	& = & \frac{1}{|\mathcal{T}(i,k)|}\sum_{M^{\prime}\in\mathcal{T}(i,k)}\frac{1}{|M^{\prime}|}=\frac{1}{k},\label{eq:notQkFsli}
	\end{eqnarray}
	where the last equality follows by Fact \ref{fact:pBG} and $G$
	is perfect. Thus,
	\begin{equation}
	Pr\left[Q_{k}\ |\ W_{i,k}\right]=1-Pr\left[\invneg Q_{k}\ |\ F_{S,\ell}^{(i)}\right]\geq1-\frac{1}{k}.\label{eq:QkW}
	\end{equation}
	Similarly, for $t\in\{k+1,\dots,\ell-1\}$ we have
	\[
	Pr\left[\invneg Q_{t}\ |\ W_{i,t}\right]\leq Pr\left[M^{(t)}[t]=r_{i}\ |\ W_{i,t}\wedge\{\text{node }r_{i}\in M^{(t)}\}\right]=\frac{1}{t},
	\]
	and therefore
	\begin{equation}
	Pr\left[Q_{t}\ |\ W_{i,t}\right]=1-Pr\left[\invneg Q_{t}\ |\ W_{i,t}\right]\geq1-\frac{1}{t}.\label{eq:QkWk}
	\end{equation}

	By combining (\ref{eq:fancy},\ref{eq:condProb},\ref{eq:QkW},\ref{eq:QkWk}),
	we obtain
	
	\begin{equation}
	Pr\left[(i,r_{i})\cup M\text{ is a matching}\ |\ F_{S,\ell}^{(i)}\right]\geq\prod_{k=\lfloor n/c\rfloor+1}^{\ell-1}\left(1-\frac{1}{k}\right)=\frac{\lfloor n/c\rfloor}{\ell-1}.\label{eq:notMatched}
	\end{equation}
	Since every summand in (\ref{eq:PEMEl}) is lower bounded by (\ref{eq:notMatched}),
	we have 
	\[
	Pr\left[E^{(\ell)}\ |\ E_{\exists}^{(\ell)}\right]\geq\frac{\lfloor n/c\rfloor}{\ell-1}.
	\]
	
\end{proof}

\subsection{Algorithm~\ref{alg:bipartite} (Omitted Proofs) \label{sub:PhaseIII}}

We now lower bound the probability that an arbitrary node $r\in R$
remains unmatched after the completion of Phase II in Algorithm~\ref{alg:bipartite}.
Our analysis uses similar arguments as in Lemma \ref{lem:PEEex},
but for the sake of completeness we present the proof below.
\begin{lemma}
	\label{lem:RinPhaseIII} For every constants $c\geq d\geq1$ and for
	every perfect bipartite graph $G=(L\cup R,E,w)$, it holds for every
	node $r\in R$ that
	\[
	Pr\left[r\text{ is not matched in }\mathrm{Phase\,II}\right]\geq\frac{d}{c}-o(1).
	\]
\end{lemma}
\begin{proof}
	Observe that
	\[
	Pr\left[r\text{ is not matched in Phase II}\right]=Pr\left[\bigwedge_{k=\lfloor n/c\rfloor+1}^{\lfloor n/d\rfloor}Q_{k}\right].
	\]
	Using (\ref{eq:notMatched}), we have
	\[
	Pr\left[r\text{ is not matched in Phase II}\right]\geq\prod_{k=\lfloor n/c\rfloor+1}^{\lfloor n/d\rfloor}\left(1-\frac{1}{k}\right)\geq\frac{\lfloor n/c\rfloor}{\lfloor n/d\rfloor}\geq\frac{d}{c}-\frac{d}{n}.
	\]
\end{proof}

%% file: appendix_deterministic_graphic_matroid.tex
\section{Deterministic Graphic Matroid Secretary Algorithm}\label{app:DetGrMatrScrAlg}
In this section, we analyze the competitive ratio of  Algorithm~\ref{alg:graphic}.
\medskip

\noindent \textbf{Theorem \ref{thm:detMatroid}.}\textit{
The deterministic Algorithm~\ref{alg:graphic} is $(1/4 - o(1))$-competitive
for the graphic matroid secretary problem.}
\medskip

The rest of this section is devoted to proving Theorem~\ref{thm:detMatroid}, and is 
organized as follows. In Subsection~\ref{sub:Elementary-Proof}, we give two useful
summation closed forms. In Subsection~\ref{app:subC2Not}, we present our notation.
In Subsection~\ref{app:StructuralLemmaDGM}, we extend Lemma~\ref{lem:Exp} to 
bipartite-matroid graphs. In Subsection~\ref{app:ProofTheoremDetMat}, we prove 
Theorem~\ref{thm:detMatroid}.

\subsection{Summation bounds} \label{sub:Elementary-Proof}
\begin{claim}
	\label{clm:Ind} For any $k\in\mathbb{N}$ and $n\in\mathbb{N}_{+}$ ,
	we have
	\[
	\sum_{\ell=n}^{n+k}\frac{1}{\ell\cdot(\ell+1)}=\frac{k+1}{n(n+k+1)}.
	\]
\end{claim}
\begin{proof}
	The proof is by induction. The base case follows by 
	\[
	\frac{1}{n}\cdot\frac{1}{n+1}+\frac{1}{n+1}\cdot\frac{1}{n+2}=\frac{2}{n(n+2)}.
	\]
	Our inductive hypothesis is $\sum_{\ell=n}^{n+k}\frac{1}{\ell\cdot(\ell+1)}=\frac{k+1}{n(n+k+1)}$.
	Then, we have
	\begin{eqnarray*}
		\sum_{\ell=n}^{n+k+1}\frac{1}{\ell}\cdot\frac{1}{\ell+1} & = & \frac{k+1}{n(n+k+1)}+\frac{1}{n+k+1}\cdot\frac{1}{n+k+2}=\frac{1}{n+k+1}\left[\frac{k+1}{n}+\frac{1}{n+k+2}\right]\\
		& = & \frac{1}{n+k+1}\left[\frac{(k+1)(n+k+1)+n+k+1}{n(n+k+2)}\right]=\frac{k+2}{n(n+k+2)}.
	\end{eqnarray*}
	
\end{proof}

\begin{claim}
	\label{clm:assympt} For any $c>1$, it holds that 
	\[
	f(c,n):=\frac{1}{n}\sum_{\ell=\lfloor n/c\rfloor}^{n-1}\frac{\left(\lfloor n/c\rfloor-1\right)\lfloor n/c\rfloor}{(\ell-1)\ell}=\frac{\lfloor n/c\rfloor}{n}\cdot\left[\frac{n-1}{n-2}-\frac{\lfloor n/c\rfloor}{n-2}\right]\apprge\frac{c-1}{c^{2}}.
	\]
	In particular, the lower bound is maximized for $c=2$ and yields
	$f(2,n)\apprge1/4$.\end{claim}
\begin{proof}
	By Claim \ref{clm:Ind}, we have
	\[
	\sum_{\ell=\lfloor n/c\rfloor}^{n-1}\frac{1}{\ell-1}\cdot\frac{1}{\ell}=\sum_{\ell=\lfloor n/c\rfloor-1}^{\lfloor n/c\rfloor+\left[n-\lfloor n/c\rfloor-2\right]}\frac{1}{\ell}\cdot\frac{1}{\ell+1}=\frac{n-\lfloor n/c\rfloor-1}{\left(\lfloor n/c\rfloor-1\right)\left(n-2\right)},
	\]
	and thus
	\begin{eqnarray*}
		\frac{1}{n}\sum_{\ell=\lfloor n/c\rfloor}^{n-1}\frac{\left(\lfloor n/c\rfloor-1\right)\lfloor n/c\rfloor}{(\ell-1)\ell} & = & \frac{\left(\lfloor n/c\rfloor-1\right)\lfloor n/c\rfloor}{n}\cdot\frac{n-\lfloor n/c\rfloor-1}{\left(\lfloor n/c\rfloor-1\right)\left(n-2\right)}\\
		& = & \frac{\lfloor n/c\rfloor}{n}\cdot\left[\frac{n-1}{n-2}-\frac{\lfloor n/c\rfloor}{n-2}\right]\\
		& \geq & \left(\frac{1}{c}-\frac{1}{n}\right)\cdot\left[1+\frac{1}{n-2}-\frac{n}{n-2}\cdot\frac{1}{c}\right]\\
		& \apprge & \frac{1}{c}\cdot\left[1-\frac{1}{c}\right]=\frac{c-1}{c^{2}}.
	\end{eqnarray*}

	Let $g(x)=(x-1)/x^{2}$. Observe that its first derivative satisfies
	\[
	\frac{d}{dx}g(x)=\frac{x^{2}-(x-1)2x}{x^{4}}=\frac{x(2-x)}{x^{4}}=0\iff x_{1}=0\quad x_{2}=2.
	\]
	Further, $g(x)$ decreases in the range $[-\infty,0]$, increases
	in $[0,2]$ and again decreases in $[2,\infty]$. Hence, we have $\max_{x>0}g(x)=g(2)=1/4$.
\end{proof}

\subsection{Notation}\label{app:subC2Not}

Given an undirected weighted graph $G^{\prime}=(V,E^{\prime},w^{\prime})$, we
construct a bipartite-matroid graph $G=(L\cup R,E,w)$ as follows:

i) let the set of the right nodes be $R=V$; ii) the set of the left nodes
be $L=E^{\prime}$, i.e., $\{u,v\}\in L$ if $\{u,v\}\in E^{\prime}$; iii)
and for each edge in $\{u,v\}\in E^{\prime}$ we insert two edges
$\{\{u,v\},u\},\{\{u,v\},v\}\in E$ with equal weight $w(\{\{u,v\},u\})=w(\{\{u,v\},v\})=w^{\prime}(\{u,v\})$.

Although $M^{(\ell)}$ is a matching, we note that $M$ is not a matchings
in the strict sense, since for an edge $\{\{u,v\},u\}\in E$ (similarly
$\{\{u,v\},v\}\in E$) to be matched it is required that both nodes
$u,v\in R$ are not yet matched. To emphasize this, we refer to $M$
as ``matching$^{\star}$''.

\subsection{Structural Lemma}\label{app:StructuralLemmaDGM}

We now extend Lemma \ref{lem:Exp} to bipartite-matroid graphs.
\begin{lemma}\label{lem:Exp-1} 
	Suppose $G=(L\cup R,E,w)$ is a perfect bipartite-matroid
	graph as above. Then, for every $c>1$ it holds for every $\ell\in\{\lfloor m/c\rfloor+1,\dots,m\}$
	that
	\[
	\mathbb{E}[A_{\ell}]\geq\frac{\lfloor m/c\rfloor-1}{(\ell-1)-1}\cdot\frac{\lfloor m/c\rfloor}{\ell-1}\cdot \frac{OPT}{m}.
	\]
	
\end{lemma}

It is straightforward to verify that the statement of Lemma~\ref{lem:BipGraph}
holds for perfect bipartite-matroid graphs, and yields
\[
	\mathbb{E}\left[w(e^{(\ell)})\ |\ E^{(\ell)}\right]\geq\frac{OPT}{m}.
\]
Hence, to prove Lemma~\ref{lem:Exp-1} it remains 
to extend the statement of Lemma \ref{lem:PEEex}.
\begin{lemma}\label{lem:pbmG} 
	For every perfect bipartite-matroid graph $G=(L\cup R,E,w)$
	it holds that
	\[
	Pr\left[E^{(\ell)}\right]\geq\frac{\lfloor m/c\rfloor-1}{(\ell-1)-1}\cdot\frac{\lfloor m/c\rfloor}{\ell-1}.
	\]
\end{lemma}
\begin{proof}
	We follow the proof in Lemma \ref{lem:PEEex}, with the amendment
	that a node $\{u_{k},v_{k}\}\in L$ and an edge $\{\{u_{k},v_{k}\},r_{k}\}\in E$.
	Recall that for a \emph{fixed }subset $S\subseteq L$ with size $\ell$
	and an edge $(i,r_{i})\in M_{S}$, we can condition on the event $F_{S,\ell}^{(i)}$
	that \{the node set of $L_{\ell}$ equals $S$\} \emph{and} \{the
	edge $e^{(\ell)}=(i,r_{i})$\}.
	
	Let $Q_{k}^{r}$ be the event that
	\[
	\{\text{node }r\not\in M^{(k)}\}\vee\left\{ \{\text{node }r\in M^{(k)}\}\wedge\{M^{(k)}[k]\neq r\}\right\} ,
	\]
	and let $\mathcal{P}_{k}$ be the event that $Q_{k}^{u_{i}}\wedge Q_{k}^{v_{i}}$.
	Then, we have
	\[
	Pr\left[\{\{u_{i},v_{i}\},r_{i}\}\cup M\text{ is a matching}^{\star}\ |\ F_{S,\ell}^{(i)}\right]=Pr\left[\bigwedge_{k=\lfloor m/c\rfloor+1}^{\ell-1}\mathcal{P}_{k}\ |\ F_{S,\ell}^{(i)}\right].
	\]
	Combining the Union bound and (\ref{eq:notQkFsli}), yields 
	\begin{eqnarray*}
		Pr\left[\invneg\mathcal{P}_{k}\ |\ F_{S,\ell}^{(i)}\right] & = & Pr\left[\invneg Q_{k}^{u_{i}}\vee\invneg Q_{k}^{v_{i}}\ |\ F_{S,\ell}^{(i)}\right]\\
		& \leq & Pr\left[\invneg Q_{k}^{u_{i}}\ |\ F_{S,\ell}^{(i)}\right]+Pr\left[\invneg Q_{k}^{v_{i}}\ |\ F_{S,\ell}^{(i)}\right]\\
		& \leq & \frac{2}{k}.
	\end{eqnarray*}
	Hence, using similar arguments as in the proof in Lemma \ref{lem:PEEex},
	we have 
	\[
	Pr\left[E^{(\ell)}\right]\geq\prod_{k=\lfloor m/c\rfloor+1}^{\ell-1}\left(1-\frac{2}{k}\right)=\frac{\lfloor m/c\rfloor-1}{(\ell-1)-1}\cdot\frac{\lfloor m/c\rfloor}{\ell-1}.
	\]
	
\end{proof}

\subsection{Proof of Theorem~\ref{thm:detMatroid}}\label{app:ProofTheoremDetMat}

Using Lemma~\ref{lem:Exp-1}, we have
\[
\mathbb{E}\left[\sum_{\ell=1}^{m}A_{\ell}\right]=\sum_{\ell=\lfloor m/c\rfloor+1}^{m}\mathbb{E}\left[A_{\ell}\right]\geq\frac{OPT}{m}\sum_{\ell=\lfloor m/c\rfloor+1}^{m}\frac{\lfloor m/c\rfloor-1}{(\ell-1)-1}\cdot\frac{\lfloor m/c\rfloor}{\ell-1}.
\]
The statement follows by Claim~\ref{clm:assympt} and noting that
\begin{equation}\label{eq:D4lowerBound}
\frac{1}{m}\sum_{\ell=\lfloor m/c\rfloor+1}^{m}\frac{\lfloor m/c\rfloor-1}{(\ell-1)-1}\cdot\frac{\lfloor m/c\rfloor}{\ell-1}
\geq\left(\frac{c-1}{c^{2}}-o(1)\right)\geq\left(\frac{1}{4}-o(1)\right).
\end{equation}

%% file: appendix_graphic_matroid_secretary_predictions.tex
\section{Graphic Matroid Secretary Algorithm with Predictions}
\label{app:graphic_prediction}

In this section, we prove the worst-case bound of $(d-1)/c^2$ in Theorem \ref{thm:graphic}, by analyzing Phase III of Algorithm \ref{alg:graphic_predic}.

\begin{theorem}\label{thm:PhaseIIIAlgo5}
	The Phase III in Algorithm~\ref{alg:graphic_predic} is
	$(\tfrac{d-1}{c^{2}}-o(1))$-competitive.
\end{theorem}

The rest of this section is denoted to proving Theorem~\ref{thm:PhaseIIIAlgo5}, 
and is organized as follows. In Subsection~\ref{app:eligibilityPairNodesProb}, 
we analyze the probability that a fixed pair of distinct vertices is eligible for
matching in Phase III. In Subsection~\ref{app:LBmatchStrEvnt}, we give a lower bound on
the event that \{$e^{(\ell)}\cup M$ is a matching$^{\star}$\}.
In Subsection~\ref{app:ConcludingArgument}, we prove Theorem~\ref{thm:PhaseIIIAlgo5}.

\subsection{Pairwise Node Eligibility in Phase III}\label{app:eligibilityPairNodesProb}

For any distinct nodes $u,v\in R$, we denote by $\Phi_{u,v}^{\notin M}$
the event that $$\{u\text{ and }v\text{ are \emph{not} matched in Phase II}\}.$$

\begin{claim}\label{clm:PhaseIII} 
	It holds that
	\[
	Pr\left[\Phi_{u,v}^{\notin M}\right]\geq\left(\frac{d}{c}\right)^{2}\cdot\frac{1-\frac{c}{m}}{1-\frac{d}{m}}.
	\]
\end{claim}
\begin{proof}
	Let $S$ be a random variable denoting the set of all nodes in $L$
	that appear in Phase I and Phase II. Let 
	\[
	e_{\max}^{\prime}(u,S)=\arg\max_{\{u,z\}\in S}w(u,z)
	\]
	be a random variable denoting the node $\{u,z\}\in L$ with largest
	weight seen in the set $S$.
	
	The proof proceeds by case distinction:
	
	\textbf{Case 1.} Suppose $e_{\max}^{\prime}(u,S)=e_{\max}^{\prime}(v,S)$,
	i.e., there is a node $\{u,v\}\in S$. Let $\mathcal{K}_{r}(S)$ be
	the event that node $e_{\max}^{\prime}(r,S)\in S$ is sampled in Phase
	I. By conditioning on the choice of $S$, we have 
	\[
	Pr\left[\Phi_{u,v}^{\notin M}\right]=\sum_{S}Pr\left[\mathcal{K}_{r}(S)\ |\ S\right]\cdot Pr\left[S\right]=\frac{m/c}{m/d}\cdot\sum_{S}Pr\left[S\right]=\frac{d}{c}.
	\]

	\textbf{Case 2.} Suppose $e_{\max}^{\prime}(u,S)\neq e_{\max}^{\prime}(v,S)$,
	i.e., there are distinct nodes $\{u,x\}\in S$ and $\{v,y\}\in S$
	with largest weight, respectively from $u$ and $v$. By conditioning
	on the choice of $S$, we have
	\[
	Pr\left[\mathcal{K}_{u}(S)\wedge\mathcal{K}_{v}(S)\ |\ S\right]=\frac{2{\frac{m}{c} \choose 2}\cdot(\frac{m}{d}-2)!}{(\frac{m}{d})!}=\left(\frac{d}{c}\right)^{2}\cdot\frac{1-\frac{c}{m}}{1-\frac{d}{m}},
	\]
	and thus 
	\[
	Pr\left[\Phi_{u,v}^{\notin M}\right]=\sum_{S}Pr\left[\mathcal{K}_{u}(S)\wedge\mathcal{K}_{v}(S)\ |\ S\right]\cdot Pr\left[S\right]=\left(\frac{d}{c}\right)^{2}\cdot\frac{1-\frac{c}{m}}{1-\frac{d}{m}}.
	\]
	
\end{proof}

\subsection{Lower Bounding the Matching$^{\star}$ Event}\label{app:LBmatchStrEvnt}

Recall that $E^{(\ell)}$ denotes the event that \{$e^{(\ell)}\cup M$ is
a matching$^{\star}$\}, where (r.v.) $M$ is the current online matching$^{\star}$, see Subsection~\ref{app:subC2Not} for details.
In order to control the possible negative side effect of selecting suboptimal edges in Phase II,
we extend Lemma~\ref{lem:pbmG} and give a lower bound on the event $E^{(\ell)}$ 
for any node $\ell\in L$ that appears in Phase III.
\begin{lemma}\label{lem:pbmG_Phases}
	For every perfect bipartite-matroid graph
	$G=(L\cup R,E,w)$, Algorithm~\ref{alg:graphic_predic} guarantees in Phase III that 
	\[
	Pr\left[E^{(\ell)}\right]\geq\frac{\lfloor m/d\rfloor-1}{(\ell-1)-1}\cdot\frac{\lfloor m/d\rfloor}{\ell-1}\cdot Pr\left[\Phi_{u,v}^{\notin M}\right],\quad\forall\ell\in\{\lfloor m/d\rfloor+1,\dots,m\}.
	\]
\end{lemma}
\begin{proof}
	We follow the proof in Lemma \ref{lem:PEEex}, with the amendment
	that a node $\{u_{k},v_{k}\}\in L$ and an edge $\{\{u_{k},v_{k}\},r_{k}\}\in E$.
	Recall that for a \emph{fixed} subset $S\subseteq L$ with size $\ell$
	and a \emph{fixed} edge $\{\{u_{i},v_{i}\},r_{i}\}\in M_{S}$, we can condition
	on the event $F_{S,\ell}^{(i)}$ that 
	$$\{\text{the set of nodes of }L_{\ell}\text{ equals }S\}\text{ and }
	\{\text{the edge }e^{(\ell)}=\{\{u_{i},v_{i}\},r_{i}\}\}.$$
	
	Let $Q_{k}^{r}$ be the event that
	\[
	\{\text{node }r\not\in M^{(k)}\}\vee\left\{ \{\text{node }r\in M^{(k)}\}\wedge\{M^{(k)}[k]\neq r\}\right\} ,
	\]
	and let $\mathcal{P}_{k}$ denotes the event $Q_{k}^{u_{i}}\wedge Q_{k}^{v_{i}}$.
	The proof proceeds by case distinction:
	
	\textbf{Case 1.} For $\ell=\lfloor m/d\rfloor+1$, we have 
	\[
	Pr\left[E^{(\ell)}\right]  =  Pr\left[\{\{u_{i},v_{i}\},r_{i}\}\cup M\text{ is a matching}^{\star}\ |\ F_{S,\ell}^{(i)}\right]
	=  Pr\left[\Phi_{u,v}^{\notin M}\right].
	\]

	\textbf{Case 2. }For $\ell=\{\lfloor m/d\rfloor+2,\dots,m\}$, we have
	\begin{eqnarray*}
		Pr\left[E^{(\ell)}\right] & = & Pr\left[\{\{u_{i},v_{i}\},r_{i}\}\cup M\text{ is a matching}^{\star}\ |\ F_{S,\ell}^{(i)}\right]\\
		& = & Pr\left[\bigwedge_{k=\lfloor m/d\rfloor+1}^{\ell-1}\mathcal{P}_{k}\ |\ F_{S,\ell}^{(i)}\wedge\Phi_{u,v}^{\notin M}\right]\cdot Pr\left[\Phi_{u,v}^{\notin M}\right].
	\end{eqnarray*}
	Combining the Union bound and (\ref{eq:notQkFsli}), yields 
	\begin{eqnarray*}
		Pr\left[\invneg\mathcal{P}_{k}\ |\ F_{S,\ell}^{(i)}\wedge\Phi_{u,v}^{\notin M}\right] & = & Pr\left[\invneg Q_{k}^{u_{i}}\vee\invneg Q_{k}^{v_{i}}\ |\ F_{S,\ell}^{(i)}\wedge\Phi_{u,v}^{\notin M}\right]\\
		& \leq & Pr\left[\invneg Q_{k}^{u_{i}}\ |\ F_{S,\ell}^{(i)}\wedge\Phi_{u,v}^{\notin M}\right]+Pr\left[\invneg Q_{k}^{v_{i}}\ |\ F_{S,\ell}^{(i)}\wedge\Phi_{u,v}^{\notin M}\right]\\
		& \leq & \frac{2}{k}.
	\end{eqnarray*}
	Hence, using similar arguments as in the proof of Lemma \ref{lem:PEEex},
	we have
	\[
	Pr\left[\bigwedge_{k=\lfloor m/d\rfloor+1}^{\ell-1}\mathcal{P}_{k}\ |\ F_{S,\ell}^{(i)}\wedge\Phi_{u,v}^{\notin 	M}\right]\geq\prod_{k=\lfloor m/d\rfloor+1}^{\ell-1}\left(1-\frac{2}{k}\right)=\frac{\lfloor m/d\rfloor-1}{(\ell-1)-1}\cdot\frac{\lfloor m/d\rfloor}{\ell-1}.
	\]
	Therefore, it holds that 
	\[
	Pr\left[E^{(\ell)}\right]\geq\frac{\lfloor m/d\rfloor-1}{(\ell-1)-1}\cdot\frac{\lfloor m/d\rfloor}{\ell-1}\cdot Pr\left[\Phi_{u,v}^{\notin M}\right].
	\]
\end{proof}

\subsection{Proof of Theorem~\ref{thm:graphic}}\label{app:ConcludingArgument}
In this section, we analyze the expected contribution in Phase III. 
Our goal now is to lower bound the expression
\[
\mathbb{E}\left[A_{\ell}\right]=\mathbb{E}\left[w(e^{(\ell)})\ |\ E^{(\ell)}\right]\cdot Pr\left[E^{(\ell)}\right].
\]

It is straightforward to verify that Lemma~\ref{lem:BipGraph} holds
in the current setting and yields
\[
	\mathbb{E}\left[w(e^{(\ell)})\ |\ E^{(\ell)}\right]\geq \frac{OPT}{m}.
\]
Further, by Lemma~\ref{lem:pbmG_Phases}, it holds for every 
$\ell\in\{\lfloor m/d\rfloor+1,\dots,m\}$ that
\begin{eqnarray*}
	\mathbb{E}\left[A_{\ell}\right] & = & \mathbb{E}\left[w(e^{(\ell)})\ |\ E^{(\ell)}\right]\cdot Pr\left[E^{(\ell)}\right]\\
	& \geq & \frac{OPT}{m}\cdot\frac{\lfloor m/d\rfloor-1}{(\ell-1)-1}\cdot\frac{\lfloor m/d\rfloor}{\ell-1}\cdot Pr\left[\Phi_{u,v}^{\notin M}\right],
\end{eqnarray*}
and thus
\[
\sum_{\ell=\lfloor m/d\rfloor+1}^{m}\mathbb{E}\left[A_{\ell}\right]\geq Pr\left[\Phi_{u,v}^{\notin M}\right]\cdot\frac{OPT}{m}\sum_{\ell=\lfloor m/d\rfloor+1}^{m}\frac{\lfloor m/d\rfloor-1}{(\ell-1)-1}\cdot\frac{\lfloor m/d\rfloor}{\ell-1}.
\]

Hence, by combining the first inequality in \eqref{eq:D4lowerBound} and Claim \ref{clm:PhaseIII},
yields
\begin{eqnarray*}
\mathbb{E}\left[\sum_{\ell=\lfloor m/d\rfloor+1}^{m}A_{\ell}\right]&\geq&\left(\frac{d}{c}\right)^{2}(1-o(1))\cdot\left(\frac{d-1}{d^{2}}-o(1)\right)\cdot OPT\\&\geq&\left(\frac{d-1}{c^{2}}-o(1)\right)\cdot OPT.
\end{eqnarray*}

%% file: truthful_transversal.tex
\section{Truthful mechanism for unit-demand domain}
\label{app:truthful}
We first re-interpret some of the terminology in Section \ref{sec:bipartite}.\footnote{We use notation and terminology closely related to that in \cite{BIKK2018}.}
We have a set $L$ of agents and a set $R$ of items. Every agent $i \in L$ has a \emph{(private) value} $v_i \geq 0$ for a set of preferred items $R_i \subseteq R$. We write $\Omega$ for the set of all (partial) matchings between $R$ and $L$, called \emph{outcomes}, and  $A_i$ for the set of all (partial) matchings in which $i$ gets matched up with some node in $R_i$, i.e., \emph{satisfying outcomes} for agent $i$.  We use $X_i : \Omega \rightarrow \{0,1\}$ as the indicator function for the set $A_i$, i.e., we have $X_i(\omega) = 1$ if and only if $\omega \in A_i$. An agent's \emph{type} is her value $v_i$, which is private information. 

In the online setting, the agents arrive in a uniformly random order,
and, upon arrival, an agent announces a (common) value $v_i'$ that she
has for the items in $R_i$. The mechanism then commits to either
choosing an outcome $\omega \in A_i$, that (partially) matches up all
agents arrived so far, or some outcome $\omega \notin A_i$ that also
(partially) matches up all agents arrived so far but not $i$, by
definition of $A_i$. It also sets a price $\rho_i$,\footnote{We use
  $\rho_i$ in order to avoid confusion with the predictions $p_r^*$.}
which depends on the choice of $\omega$. We assume that $\rho_i\in\mathbb{N}_0$. The outcome $\omega$ should be consistent with previous outcomes in the sense that the proposed matching in step $i - 1$ is a submatching of that committed to in step $i$.\footnote{This submatching corresponds to the online matching $M$ that we construct in Algorithm \ref{alg:bipartite}.} After all agents have arrived the mechanism commits to the final outcome $\omega$ in step $n$. Informally speaking, whenever an agent arrives, the mechanism either offers her an item $r \in R_i$ at some price $p_i$, or it offers her nothing.

The goal of an agent is to maximize her utility 
$$
u_i(\omega) = v_i X_i(\omega) - \rho_i,
$$
and the goal of the mechanism designer is to maximize the \emph{social welfare} 
$$
\sum_{i \in L} v_i X_i(\omega).
$$

We will show that Algorithm \ref{sec:bipartite} can be used to design a truthful (online) mechanism in the case of so-called \emph{single-value unit-demand domains}. 

We want to design a mechanism that is \emph{truthful}, meaning that it is always in the best interest of an agent to report $v_i' = v_i$ for any declarations of the other agents and any arrival order of the agents. We do this by showing that the procedures in Phase II and III can be modified so that the assignment rules in Algorithm \ref{alg:bipartite} become \emph{monotone}, see, e.g., Chapter 16 in \cite{NRTV2007}. It is not hard to see, and well-known, that this kind of monotonicity incentivizes truthfulness. \\

\noindent In order to turn Algorithm \ref{alg:bipartite} into a truthful mechanism we need to make some modifications, and define a pricing scheme. The outcome $\omega$ in every step corresponds to the (online) matching $M$ we have chosen so far, together with the  edge $e^{i}$ in case we match up $i$ when she arrives.

The modification for Phase I is straightforward. For every arriving agent $i$, we choose as outcome the empty matching (which is not in $A_i$ for any $i$), and set a price of $\rho_i = 0$. That is, we assign no item to any agent arriving in Phase I. 

For Phase III we can use a similar idea as the \emph{price sampling algorithm} in \cite{BIKK2018}. We set price-thresholds $\rho(r) = p_r^* - \lambda$. Whenever an agent arrives in Phase III, we look at all the items in $R_i$ for which her reported value $v_i'$ exceeds the price threshold $\rho(r)$ and assign her (if any) the item $r'$ for which the price threshold is the lowest among those. We charge her a price of $\rho_i = p_{r'}^* - \lambda$ for the item $r'$. The fact that this incentivizes truthful behavior follows from similar arguments as those given in \cite{BIKK2018}. 

In order to incentivize truthfulness in Phase II, we exploit the fact that we are free to choose any fixed algorithm to compute the (offline) matching $M^i$ in that phase once a node $i$ has arrived. That is, Algorithm \ref{alg:bipartite} works for every choice of such an offline algorithm, but in order to guarantee truthfulness, we need a specific type of bipartite matching algorithm.\\ 

\noindent We first introduce an additional graph-theoretical
definition. We say that an instance of the bipartite matching problem
on a (bipartite) graph $G = (A \cup B, E)$, with $|A| = n$, has
uniform edge weights if all edges adjacent to a given node $a \in A$ have the same weight, i.e., we have $w(a,b) = w(a,b')$ for all $b,b' \in \mathcal{N}(a) \subseteq B$. We denote this common weight by $w_a$. Moreover, we write $(w_a',w_{-a})$ to denote the vector  $(w_1,\dots,w_{a-1},w_a',w_{a+1},\dots,w_n)$ in which $w_a$ is replaced by $w_a'$.

\begin{definition}
  \label{def:monotone}
We say that a deterministic (bipartite matching) algorithm
$\mathcal{A}$ is \emph{monotone for instances with uniform edge weights} if the following holds. For an instance $I = (G, (w_1,\dots,w_n))$ and  any fixed $a \in A$, there exists a critical value $\tau_a = \tau_a(I)$ such that $\mathcal{A}$ has the following properties:
\begin{enumerate}
\item The node $a$ does not get matched for any $ I = (G,(w_a',w_{-a}))$ with $w_a' < \tau_a$;
\item There exists a node $b \in \mathcal{N}(a)$, such that for any $w_a' \geq \tau_a$ the node $a$ gets matched up to $b$ in $I = (G,(w_a',w_{-a}))$.
\end{enumerate}
\end{definition}
We emphasize that whenever $w_a' \geq \tau_a$, we want $a$ to be matched up to the \emph{same} node $b$. So for any $w_a'$, the node $a$ either does not get matched up by $\mathcal{A}$, or it gets matched up to some fixed $b$. \\

\noindent Now, in Phase II we compute an offline matching $M^i$ using algorithm $\mathcal{A}$. If the edge $e^i = (i,r)$ assigned to $i$ in $M^i$ (if any) satisfies the condition that $M \cup e^i$ is a matching, where $M$ is the online matching constructed so far, we assign item $r$ to $i$. We charge a price of $\rho_i = \tau_a$.
It is not hard to see, using a standard argument, that any monotone bipartite matching algorithm with the given pricing rule incentivizes truthful behaviour (this is left to the reader).
 
A summary of the modifications to Algorithm \ref{alg:bipartite} can be found in Algorithm \ref{alg:bipartite_truthful}.
\IncMargin{1em}
\begin{algorithm}
\SetKwData{Left}{left}\SetKwData{This}{this}\SetKwData{Up}{up}
\SetKwFunction{Union}{Union}\SetKwFunction{FindCompress}{FindCompress}
\SetKwInOut{Input}{Input}\SetKwInOut{Output}{Output}
\Input{Predictions $p^*= (p_1^*,\dots,p_{|R|}^*)$, confidence parameter $\lambda > 0$, and $c > d \geq 1$. Monotone bipartite matching algorithm $\mathcal{A}$.}
\Output{Matching (or assignment) $M$ and prices $\rho = (\rho_1,\dots,\rho_n)$.}\medskip

\textbf{Phase I:} \\
\For{$i = 1,\dots,\lfloor n/c \rfloor$}{
Assign no item to agent $i$ and set $\rho_i = 0$.}
Let $L' = \{\ell_1,\dots,\ell_{\lfloor n/c \rfloor}\}$ and $M = \emptyset$.\\
\textbf{Phase II:} \\
\For{$i = \lfloor n/c \rfloor + 1, \dots, \lfloor n/d \rfloor$}{
Let $L' = L' \cup \{i\}$.\\
Let $M^{i} = \text{optimal matching on } G[L' \cup R]$ computed using $\mathcal{A}$. \\
Let $e^{i} = (i,r)$ be the edge assigned to $i$ in $M^{i}$.\\
\If{$M \cup e^i$ is a matching}{
Set $M = M \cup \{e^i\}$, i.e., assign item $r$ to $i$.\\
Set $\rho_i = \tau_i(G[L' \cup R],(v_1',\dots,v_i'))$.}}
\textbf{Phase III:} \\
\For{$i =  \lfloor n/d \rfloor + 1,\dots,n$}{
Let $S = \{r  \in \mathcal{N}(i) : v_i' \geq p_r^* - \lambda \text{ and } r \notin R[M]\}$\\
\If{$S \neq \emptyset$}{
Set $M = M \cup \{i,r'\}$ where $r' = \text{argmax} \{v_i' - (p_r^* - \lambda) : r \in S\}$, i.e., assign item $r'$ to agent $i$.\\
Set $\rho_i = p_{r'}^* - \lambda$.
}}
\caption{Truthful online mechanism for single value unit-demand  domains}
\label{alg:bipartite_truthful}
\end{algorithm}
\DecMargin{1em}
Based on Algorithm \ref{alg:bipartite_truthful}, and the analysis of
Algorithm \ref{alg:bipartite}, we obtain the following theorem,
provided there exists a monotone bipartite matching algorithm
$\mathcal{A}$ for instances with uniform edge weights. We will show the existence of such an algorithm in the proof of Theorem \ref{thm:truthful}.

\begin{theorem}
There is a deterministic truthful mechanism that is a
$g_{\lambda,c,d}(\eta)$-approximation for the problem of social
welfare maximization, with $g_{\lambda,c,d}$ as in Theorem
\ref{thm:bipartite} \label{thm:truthful}. Furthermore, the alignment rule and pricing scheme of the mechanism are computable in polynomial time.
\end{theorem}
\begin{proof}
As mentioned above, it suffices to show the existence of a monotone bipartite matching algorithm. 
Fix an arbitrary total order $\succ$ over all the elements of $L\times
R$. We say that a matching $M=\{m_1,m_2\dots\}$ is lexicographically larger than
matching $M'=\{m_1',m_2',\dots\}$ with $M'\neq M$ and write
$M\succ_{\text{lex}}M'$ if and only if, there exists an integer $k \geq 0$ such that
$m_i=m_i'$ for $i=1,\dots k$ and $m_{k+1}\succ m_{k+1}$.

We claim that any existing exact algorithm for the maximum weight
bipartite matching problem can be easily converted to an exact
algorithm for the \emph{lexicographically maximum weight bipartite
  matching problem}, i.e., the problem where one seeks to find
lexicographically largest matching among the maximum weight
matchings. Indeed, consider some existing algorithm $\mathcal{A'}$ for
maximum weight bipartite matching, and assume that $\mathcal{A'}$ on
instance $G = (L\cup R, E)$ gives a matching of cost $\text{OPT}$. Now
consider the maximum $e=(\ell,r)\in E$ according to $\succ$, and
compute the maximum weight bipartite matching in
$G'=(\{L\setminus\{u\}\}\cup\{R\setminus\{v\}\},E\setminus\{e': e'\cap
u\neq\emptyset \text{ or } e'\cap v \neq \emptyset\})$. If the
resulting solution has cost $< \text{OPT} - w_e$, then we know that
$e$ is not in any maximum weight bipartite matching, and can
recursively continue with $\text{OPT}$ and $G:=(L\cup R,
E\setminus\{e\})$. If on the other hand the resulting solution has
cost $= \text{OPT}-w_e$, then $e$ is part of some maximum weight
bipartite matching (and in particular the lexicographically maximal
one), so we fix $e$ and compute the rest of the matching recursively
with $\text{OPT} := \text{OPT}- w_e$ and $G := G'$. By construction,
we will in the end obtain the lexicographically largest maximum weight
bipartite matching. We note that the lexicographically maximum
weighted bipartite matching is unique for any input
instance. Furthermore, we run algorithm $\mathcal{A'}$ at most
$O(n^2)$ times. In other words we can determine such a matching in
polynomial time for any given set of weights.

 It remains therefore to show that the value $\tau_a$, as described in
Definition~\ref{def:monotone}, exists. It is easy to see that there
exists a value $w_a'$ for which $a$ is matched in any
lexicographically maximum weighted bipartite matching (for example,
set $w_a'=\sum_{e\in E\setminus {a\cap E}} w_e + \epsilon =:W$). Therefore
it suffices to show that if $a$ is matched to some vertex
$b\in\mathcal{N}(a)$ for some weight $w_a' = \tau_a$, then it is also
matched to the same neighbor $b$ for any weight $w_a''
>\tau_a$. Assume for the sake of contradiction that this is not the
case, and let the corresponding two lexicographically maximum weighted
matchings be $M'\ni (a,b)$ and $M''$. Note that $M''$ could leave $a$
unmatched, or it could contain an edge adjacent to $a$ but not
$(a,b)$. In the following we use $\text{value}(M)$ to describe the
total weight of matching $M$, i.e., $\sum_{e\in M} w_e$.We also note that, since $M'$ is a lexicographically maximal
weighted bipartite matching for
$(w_a',w_{-a})$, we have that either
$\text{value}(M'(w_a',w_{-a}))>\text{value}(M''(w_a',w_{-a}))$ or the
two values are equal and $M'(w_a',w_{-a}) \succ_{\text{lex}}
M''(w_a',w_{-a})$.

For the second case, and since the lexicographic order does not depend on
the edge weights, we also have $M'(w_a'',w_{-a}) \succ_{\text{lex}}
M''(w_a'',w_{-a})$. Therefore, for $M''$ to be a lexicographically
maximal weighted bipartite matching for $(w_a'',w_{-a})$, it has to be
that 
$$\text{value}(M''(w_a'',w_{-a})) >
\text{value}(M'(w_a'',w_{-a})).
$$ 
However, this is not possible, since
$\text{value}(M'(w_a'',w_{-a})) = \text{value}(M'(w_a',w_{-a})) +
w_a''-w_a'$, and 
$$\text{value}(M''(w_a'',w_{-a})) \le
\text{value}(M''(w_a',w_{-a})) + w_a''-w_a',$$
which is a contradiction. For the first case:
\begin{align*}
\text{value}(M'(w_a'',w_{-a})) &= \text{value}(M'(w_a',w_{-a})) +
w_a''-w_a' \\ 
& >\text{value}(M''(w_a',w_{-a})) + w_a''-w_a' \\
& \ge \text{value}(M''(w_a'',w_{-a})),
\end{align*}
This gives a contradiction.
This concludes the proof for the first statement of the theorem. 

For the second part, first recall that we run some deterministic
algorithm for the bipartite weighted maximum matching $O(n^2)$ many
times. It remains to show that $\tau_a$ can be computed. We know that
$\tau_a$ exists and is in the range $[0,W]$. Also recall that $\tau_a$
is the smallest value for $w_a'$ such that $a$ gets matched in a
lexicographically maximal weighted bipartite matching for
$(w_a',w_{-a})$. We therefore can perform a binary search over
$[0,W]$, by trying out if in the corresponding instance $a$ is part of
the lexicographically maximum weighted bipartite matching.  Since we
assume weights to be integers, this will terminate after at most $\log
W$ many steps. By using for example the Edmonds-Karp algorithm for
$\mathcal{A'}$, our mechanism has a running time of $O(n^4m^2\log W)$
which is polynomial in the input size.
\end{proof}